\documentclass[11pt]{article}
\usepackage{amsmath,amssymb,amsthm}
\usepackage{graphics}
\usepackage{tikz}
\usepackage{epic,eepic}
\usepackage{bm}
\usepackage[affil-it]{authblk}
\usepackage{url,enumitem}

\theoremstyle{plain}
\newtheorem{theorem}{Theorem}[section]
\newtheorem{lemma}[theorem]{Lemma}

\newtheorem{proposition}[theorem]{Proposition}
\newtheorem{corollary}[theorem]{Corollary}

\theoremstyle{definition}
\newtheorem{definition}[theorem]{Definition}
\newtheorem{remark}[theorem]{Remark}

\newcommand{\ua}{\mathord{\uparrow}}
\newcommand{\da}{\mathord{\downarrow}}
\newcommand{\rom}[1]{\rm{\uppercase\expandafter{\romannumeral #1}}}

\newcommand{\set}[2]{\{#1\mid#2\}}
\newcommand{\oneset}[1]{\{#1\}}

\newcommand{\R}{\overline{\mathbb R}_{+} }

\newcommand{\VV}{{\mathcal V}}
\newcommand{\J}{{\mathcal J}}
\newcommand{\N}{{\mathbb N}}
\newcommand{\Ncof}{{\mathbb N}_{\rm cof}}
\newcommand{\dcpo}{\mathbf{DCPO}}
\renewcommand{\O}{{\mathcal O}}
\newcommand{\SSS}{{\mathcal S}}
\newcommand{\MM}{{\mathcal M}}

\newcommand{\Val}{{\mathrm{Val}}}

\newcommand\upc{\ua}
\newcommand\dc{\da}
\newcommand{\real}{\mathbb{R}}
\newcommand\Rl{\real_\ell}
\newcommand\Rp{\real_+}
\newcommand\creal{\overline{\real}_+}

\newcommand\directed{\sideset{}{^{\,\makebox[0pt]{$\scriptstyle\uparrow$}\;}}}
\newcommand\filtered{\sideset{}{^{\,\makebox[0pt]{$\scriptstyle\downarrow$}\;}}}
\newcommand\dsup{\directed\sup}
\newcommand\finf{\filtered\inf}

\newcommand\fcap{\filtered\bigcap}
\newcommand\eqdef{\mathrel{\buildrel \text{def}\over=}}
\newcommand\diff{\smallsetminus}
\newcommand\Max{\mathop{\mathrm{Max}}}
\newcommand\Smyth{\mathcal{Q}}
\newcommand\img{\mathop{\mathrm{Im}}}

\newcommand\egame[1]{\mathfrak{e}_{#1}}

\title{Separating minimal valuations, point-continuous valuations and continuous valuations}

\author[1]{Jean Goubault-Larrecq
    \thanks{Email: \texttt{goubault@lsv.fr}}}
\author[2]{Xiaodong Jia
    \thanks{Email (corresponding author): \texttt{jiaxiaodong@hnu.edu.cn}}}

\affil[1]{Universit\'e Paris-Saclay, CNRS, ENS Paris-Saclay, Laboratoire M\'ethodes Formelles, 91190, Gif-sur-Yvette, France}
\affil[2]{School of Mathematics, Hunan University, Changsha, Hunan, 410082, China}

\begin{document}
\maketitle

\begin{abstract}

We give two concrete examples of continuous valuations on dcpo's to separate
minimal valuations, point-continuous valuations and continuous valuations:

\begin{enumerate}

\item  Let $\J$ be the Johnstone's non-sober dcpo, and $\mu$ be the continuous 
valuation on $\J$ with $\mu(U) =1$ for nonempty Scott opens $U$ and 
$\mu(U) = 0$ for $U=\emptyset$. Then $\mu$ is a point-continuous valuation on
$\J$ that is not minimal.

\item  Lebesgue measure extends to a measure on the Sorgenfrey line~$\Rl$.
  Its restriction to the open subsets of $\Rl$ is a continuous
  valuation~$\lambda$.  Then its image valuation
  $\overline\lambda$ through the embedding of $\Rl$ into its Smyth
  powerdomain $\Smyth\Rl$ in the Scott topology is a continuous valuation that is not
  point-continuous.
  \end{enumerate} 
 We believe that our construction $\overline\lambda$ might be useful in 
  giving counterexamples displaying the failure of the general Fubini-type equations on 
  dcpo's.
\end{abstract}

\section{Introduction}
\label{sec:introduction}

Continuous valuations on topological spaces are analogues of measures 
on measurable spaces. In domain theory, continuous valuations on dcpo's 
with the Scott topology are employed by computer scientists and 
mathematicians to give denotational meanings to probabilistic programming
 languages. This line of work dates back to Jones and Plotkin~\cite{jones89, jones90}. 
 Indeed, in her Ph.D. thesis, Jones developed the theory of valuations and used the valuations 
 monad $\VV$ on the category $\dcpo$ of dcpo's and Scott-continuous maps to 
 give denotational semantics to probabilistic programming languages. 

While the valuations monad on the category $\dcpo$ enjoys many nice 
properties, for example this monad is a strong monad and dcpo-enriched, 
it is unknown whether it is a commutative monad on the same category. 
As a result, it would be difficult, using the valuations monad,  to establish 
the so-called contextual equivalence between programs that only differ 
in the order of sampling random variables. To combat this problem, 
the authors in \cite{monad-m} constructed submonads of the valuations monad $\VV$
that are commutative on the category of dcpo's. Among their construction, 
there is a least submonad of the valuations monad that consists of which we 
call \emph{minimal valuations}. Minimal valuations are these continuous 
valuations that are in the d-closure of the simple valuations; precisely, they 
consist of directed suprema of simple valuations, directed suprema of directed 
suprema of simple valuations and so forth, transfinitely. Every minimal valuation 
is a point-continuous valuation in the sense of Heckmann~\cite{heckmann96, monad-m}. 
Heckmann\cite{heckmann96} proved that the class of 
point-continuous valuations on space $X$ form the sobrification of the space of simple
valuations on $X$, both in the so-called weak topology.
Every point-continuous valuation is a continuous valuation~\cite[Proposition~3.1]{heckmann95}.

It is relatively easy to see that on general topological spaces minimal
valuations form a strictly smaller class than that of point-continuous valuations.
However, it is unknown whether the same is true on dcpo's with the Scott topology. 
The first example in this paper clarifies the difference between minimal valuations and 
point-continuous valuations on dcpo's. Concretely,  we consider the 
well-known Johnstone's non-sober dcpo $ \J$ and the ``constant-1 valuation'' 
$\mu$ on $\J$  defined by $\mu (U) = 1$ if $U$ is nonempty and $\mu(\emptyset) = 0$. 
We show that every bounded continuous valuation on $\J$ can be written as a sum 
of some discrete valuation and a scalar multiple of $\mu$. This enables us to 
conclude that every continuous valuation on $\J$ is actually point-continuous. 
Moreover, we prove that the continuous valuation~$\mu$ is not in the d-closure 
of simple valuations, hence it serves as an example that separates minimal valuations 
from point-continuous valuations. This example is included in Section~\ref{sec:johnstone}.

Similar to the difference between minimal valuations and point-continuous valuations, continuous valuations 
that are not point-continuous can be easily found on topological spaces. For example, Lebesgue
measure, restricted to the usual opens of reals, is a continuous valuation that
is not point-continuous.  However,  it has been unknown whether point-continuous valuations differ from 
continuous valuation on dcpo's since 1996.
The second goal of this note is to give an example of a continuous valuation on a \emph{dcpo} that is not point-continuous. 
In order to find such an example, one is tempted to find the simplest
possible example, and typically to find a continuous valuation that
takes only two values, $0$ or $1$, and hoping that it would not be
point-continuous.  However, we notice that such a strategy
cannot work, as we will see in Section~\ref{sec:any-non-point}.
Hence, we will have to work a bit more.  We show how one can build
certain continuous valuation on the Sorgenfrey line $\Rl$ in
Section~\ref{sec:sorgenfrey}, including one based on Lebesgue measure $\lambda$.
We study the compact subsets of $\Rl$ in
Section~\ref{sec:compact-subsets-rl}, as a preparation to studying the
dcpo $\Smyth\Rl$ of compact subsets of $\Rl$ under reverse inclusion,
and showing that the natural map from $\Rl$ to $\Smyth\Rl$ is a
subspace embedding in Section~\ref{sec:dcpo-smythrl}.  We transport
Lebesgue measure $\lambda$ along this embedding, and we will show that the
resulting continuous valuation  $\overline\lambda$ on $\Smyth\Rl$ fails to be
point-continuous in Section~\ref{sec:continuous-non-point}.  We believe that our construction $\overline\lambda$ might be useful in 
  giving counterexamples displaying the failure of the general Fubini-type equations on 
  dcpo's, which is a longstanding open problem in domain theory. More detailed 
  discussion about this part is included in the concluding remarks.

\section{Preliminaries}

We use standard concepts and notations from topology, measure theory, 
and domain theory. The reader is referred to \cite{abramsky94, goubault13, gierz03}
for topology and domain theory, and to \cite{Royden88} for measure theory. 

\subsection{Valuations}

On a topological space $X$, a valuation $\nu$ is a  map 
from the set $\mathcal OX$  of opens of $X$ to the extended reals~$\R$, satisfying
 \emph{strictness} $(\nu(\emptyset) = 0)$, \emph{monotonicity} 
 $(U\subseteq V \Rightarrow  \nu(U) \leq \nu(V))$ and \emph{modularity} 
 $(\nu(U) + \nu(V) = \nu(U\cup V) + \nu(U\cap V))$. A valuation $\nu$ on $X$
 is called \emph{continuous} if it is Scott-continuous from $\mathcal OX$ to $\R$,
 and  it is called \emph{bounded} if $\nu(X) < \infty$.  Continuous valuations 
 are ordered in the \emph{stochastic order}: $\nu_{1} \leq \nu_{2}$ if and 
 only if $\nu_{1}(U) \leq \nu_{2}(U)$ for all opens $U$ of~$X$.  The set of 
 all continuous valuations on~$X$, which we denote as $\VV X$, is a dcpo 
 in the stochastic order. 
Canonical examples of continuous valuations on $X$ include 
\emph{Dirac valuations} $\delta_{x}$ for $x\in X$, where 
$\delta_{x}(U) = 1$ if $x\in U$ and $\delta_{x}(U) = 0$ if $x\notin U$.
As a dcpo, $\VV X$ is closed under suprema, it also is closed under scalar multiplication and sum, 
 for $\nu_{i}\in \VV X$ and $r_{i}\in [0, \infty[$, $i =1, \cdots, n$, 
 the sum  $\sum_{i=1}^{n} r_{i}\nu_{i}$ which is defined by 
 $(\sum_{i=1}^{n} r_{i}\nu_{i}) (U) = \sum_{i=1}^{n} r_{i}\nu_{i}(U)$ 
 also is in $\VV X$. For $r_{i}\in [0, \infty[$ and $x_{i}\in X$, $i=1, \cdots, n$, 
 the finite sum $\sum_{i=1}^{n} r_{i}\delta_{x_{i}}$ is called a 
\emph{simple valuation} on $X$. 
The set of all simple valuations on $X$ is denoted by $\SSS X$. 
Valuations of the form 
$\sum_{i}^{\infty}r_{i}\delta_{x_{i}} = \sup_{n\in \mathbb N} \sum_{i=1}^{n}r_{i}\delta_{x_{i}}$
are called \emph{discrete valuations}. 
The smallest sub-dcpo of $\VV X$ that contains $\SSS X$ (hence all discrete valuations on $X$)
 is denoted by $\MM X$, and every valuations in $\MM X$ is called \emph{minimal}. It is 
easy to see that for each minimal valuation $\nu$, $\nu$ is either a simple
 valuation, or a directed supremum of simple valuations, or a directed
 supremum of directed suprema of simple valuations $\cdots$, transfinitely. 
A valuation $\nu$ on a 
 space $X$ is \emph{point-continuous} if and only if
for every open subset $U$, for every real number $r$ such that
$0 \leq r < \nu (U)$, there is a finite subset $A$ of $U$ such that
$\nu (V) > r$ for every open neighborhood $V$ of $A$. 
Minimal valuations are point-continuous~\cite{monad-m}, and 
point-continuous valuations are continuous valuations~\cite[Proposition~3.1]{heckmann95}. 

\subsection{Ring of sets}

We will need the notion of \emph{Boolean ring of sets} (\emph{ring of sets} for short). On a 
set $X$, a ring of sets on $X$ is a lattice of sets consisting of subsets of $X$ that also is closed
under relative complements. For a topological space $X$, the set $\O X$ of all opens of $X$
is a lattice of sets, and the ring of sets \emph{generated} by $\O X$ is the intersection of all 
rings of sets on $X$ that contain $\O X$, and it is denoted by $\mathcal A(\O X)$.

\begin{lemma}{\rm \cite[Lemma IV-9.2]{gierz03}}
Let $X$ be a topological space. For each set $A$ in $\mathcal A(\O X)$, 
the ring of sets generated by open sets of $X$, $A$ is of the form of
 a finite disjoint union $ \coprod_{i=1}^{n} U_{i}\setminus V_{i}  $, 
 where $U_{i}$ and $V_{i}$ are open subsets of~$X$. One can 
 also stipulate that $V_{i}\subseteq U_{i}$ for each $i$.
\end{lemma}

For open subsets $U$ and $V$ of $X$, the set difference $U\setminus V$ is called a \emph{crescent}. 
Since $U\setminus V = U\setminus (U\cap V)$,
in the sequel, when we write a set $A\in \mathcal A(\O X)$ as $\coprod_{i=1}^{n} U_{i}\setminus V_{i}$,
we always assume that $V_{i}\subseteq U_{i}$ for each $i$.

\begin{lemma}{\rm \cite[Section 3.3]{heckmann96}}
\label{lemma:restritoring}
Let $X$ be a topological space and $\mu$ be a bounded continuous 
valuation on $X$.  For each set $A= \coprod_{i=1}^{n} U_{i}\setminus V_{i}$
 in $\mathcal A(\O X)$, define $\mu_{A} = \sum_{i=1}^{n}  \mu|_{U_{i}}^{V_{i} }$, 
 where for open sets $U, V$ and $W$, 
 $\mu| _{U}^{V}(W) = \mu(W\cap U)-\mu(W\cap V\cap U) $. 
 Then $\mu_{A}$ is a bounded continuous valuation for each 
 $A\in \mathcal A(\O X)$. In particular, $\mu_{U\setminus V} = \mu| _{U}^{V}$. 
 Note that $\mu_{U} = \mu|_{U}^{\emptyset}$ for every open subset $U$.
\end{lemma}

Note that it is possible that for each $A\in \mathcal A(\O X)$, 
$A$ can be written as a disjoint union $\coprod_{i=1}^{n} U_{i}\setminus V_{i}$ 
or $\coprod_{j=1}^{m} U_{j}\setminus V_{j}$. However, 
when $\coprod_{i=1}^{n} U_{i}\setminus V_{i} = \coprod_{i=j}^{m} U_{j}\setminus V_{j}$,  
we will always have that 
$\sum_{i=1}^{n}  \mu|_{U_{i}}^{V_{i} }= \sum_{j=1}^{m}  \mu|_{U_{j}}^{V_{j} }$. 
This validates the definition of $\mu_{A}$ in the previous lemma. 

\begin{lemma}\label{sum-of-musub}
For two disjoint sets $A, B\in \mathcal A(\O X)$, 
$\mu_{A\cup B} = \mu_{A}+\mu_{B}$. 
This implies that $\mu_{A}\leq \mu_{B}$ when $A\subseteq B$.
\end{lemma}
\begin{proof}
From the above remark and straightforward computation. 
\end{proof}

\begin{lemma}\label{restriction-on-points}
Let $X$ be a $T_{0}$ topological space and $\mu$ be a 
bounded continuous valuation on $X$. If $\oneset{a}$ is
in $\mathcal A(\O X)$, then there exist opens $U, V$ 
with $V\subseteq U$ and $U\setminus V= \oneset{a}$. 
Moreover, in this case $\mu_{\{a\}}= r_{a} \delta_{a}$, 
where $r_{a} = \mu_{\{a\}}(X)= \mu(U)-\mu(V)=\mu_{\{a\}}(U)$.
\end{lemma}
\begin{proof}
The first assertion is obvious. 

For the second assertion, 
since $X$ is $T_{0}$, we only need to prove that for each 
open $O$, $\mu_{a}(O)$ is equal to $r_{a}$ if $a\in O$ 
and to $0$ if $a\notin O$. Assume that $U$ and $V$ are 
open subsets of $X$ with $V\subseteq U$ and $U\setminus V = \{a\}$. 
For an open subset $O$, if $a\in O$,  
then $O\cap U\setminus O\cap V= U\setminus V=\oneset{a}$, 
which implies that $\mu_{a}(O) = \mu|_{U}^{V}(O) = \mu|_{ U\cap O}^{ V\cap O }(O)
= \mu|_{ U\cap O}^{ V\cap O }(U) = \mu_{a}(U)$. 
If $a$ is not in $O$, then $O\cap U = O\cap V$. Hence $\mu_{a}(O)= 0$.
\end{proof}

\section{Point-continuous valuations need not be minimal valuations}
\label{sec:johnstone}

In this section, we built a point-continuous valuation~$\mu$ on the well-known Johnstone's
non-sober dcpo, which is not a minimal valuation.

\subsection{Valuations on $\Ncof$}

Let $\Ncof$ be the topological space of natural numbers equipped with 
the co-finite topology. It is easy to verity that the the map 
$\beta \colon \O \Ncof \to \R$ defined by
$$ \beta(U)  = \begin{cases}
         1, & \text{$U \subseteq \Ncof$ is open and  nonempty}; \\
         0, & \text{$U$ is empty}         
        \end{cases}$$
is a bounded continuous valuation on $\Ncof$. 

\begin{lemma}\label{point-is-crescent}
For each $i\in \Ncof$, the set $\oneset{i}$ is closed. Hence it is in the ring 
of sets generated by the co-finite topology on $\Ncof$. 
\end{lemma}
\begin{proof}
Straightforward.
\end{proof}

\begin{proposition}\label{valuations-on-N}
Let $\nu$ be a bounded continuous valuation on $\Ncof$. Then there exists a 
discrete valuation $\alpha$ on $\Ncof$ and nonnegative real number 
$r$ such that $\nu = \alpha + r\beta$. 
\end{proposition}
\begin{proof}
We let $\alpha = \sum_{i\in \N} \nu_{\{i\}}$. By Lemma~\ref{point-is-crescent}
 and  Lemma~\ref{lemma:restritoring}, each $\nu_{\{ i\}}$ is a continuous valuation. 
Since $\Ncof$ is a $T_{0}$ topological space, it follows from 
Lemma~\ref{restriction-on-points} that $\alpha$ is a discrete valuation. 
Now we define the map 
$$\nu^{*} \colon \O \Ncof \to \mathbb R_{+} :: U\mapsto \nu (U) - \alpha (U).$$ 
We proceed to show that $\nu^{*}$ is a multiple of $\beta$, that is, 
there exists some $r\in [0,\infty[$ such that $\nu^{*} = r \beta$. 

First, for each $n\in \N$, 
$\nu(U) - \sum_{i=1}^{n} \nu_{\{i\}}(U)  = \nu_{ \N\setminus \{1, 2, ..., n \}}(U)$ 
is nonnegative, hence $\nu^{*}(U)$, which is the infimum of 
$\nu(U) - \sum_{i=1}^{n} \nu_{\{i\}}(U), n\in \N$, 
indeed takes values in $\mathbb R_{+}$.

Second, for nonempty open sets $U$ and $V$ with $V\subseteq U$, 
we prove that $\nu^{*}(U) = \nu^{*}(V)$. Because $V$ is co-finite, 
we know that $U\setminus V$ is a finite set, which we denote by $F$. 
Then we know 
\begin{align*}
\nu^{*}(U) & = \nu(U) - \sum_{i= 1}^{\infty} \nu_{\{i\}}(U)  							& \text{definition of $\nu^{*}$} \\
& = \nu_{U}(U) - \sum_{i= 1}^{\infty} \nu_{\{i\}}(U) 									& \text{$\nu_{U}(U) = \nu(U\cap U)$} \\
& = \nu_{V}(U) + \sum_{i\in F}\nu_{\{i\}}(U)  - \sum_{i= 1}^{\infty} \nu_{\{i\}}(U)		&\text{by Lemma~\ref{sum-of-musub} and $V\cup F = U$} \\
& = \nu_{V}(V)  - \sum_{i= 1}^{\infty} \nu_{\{i\}}(V)								&\text{by Lemma~\ref{restriction-on-points}} \\
& = \nu^{*}(V).															& \text{definition of $\nu^{*}$}
\end{align*}
Now for general nonempty opens $U$ and $V$, we have 
$\nu^{*}(U) = \nu^{*}(U\cap V) = \nu^{*}(V)$. Let $r = \nu^{*}(X)$, 
then we know that $r$ is nonnegative from above and that 
$\nu^{*} = r\beta$ (hence $\nu^{*}$ also is a continuous valuation).

Finally, we conclude the proof by the fact that 
$\nu = \alpha + \nu^{*} = \alpha + r \beta$. 
\end{proof}

\subsection{Valuations on Johnstone's non-sober dcpo $\J$}
\label{section-non-soberj}

In 1980, Johnstone gave the first dcpo which is not sober
in the Scott topology~\cite{johnstone81}. This dcpo, which we denote by $\J$, 
serves as a basic building block in several   counterexamples 
in domain theory~\cite{isbell82, ho18}. In this subsection, 
we will use it to construct a continuous valuation that is not minimal.

\begin{definition}[The dcpo $\J$] 
\label{def:johnstone}
Let $\mathbb N$ be the set of natural numbers and  $\J = \mathbb N\times (\mathbb N \cup \{\infty\} )$. 
The order on $\J$ is defined by $(a, b) \leq (c, d)$ if and only if either 
$a = c$ and $b\leq d$, or $d = \infty$ and $ b\leq c$. 
\end{definition}

The structure of $\J$ is depicted in Figure~\ref{fig:1}. 

\begin{figure}[h]
 \begin{center}
\includegraphics[width=.5\linewidth]{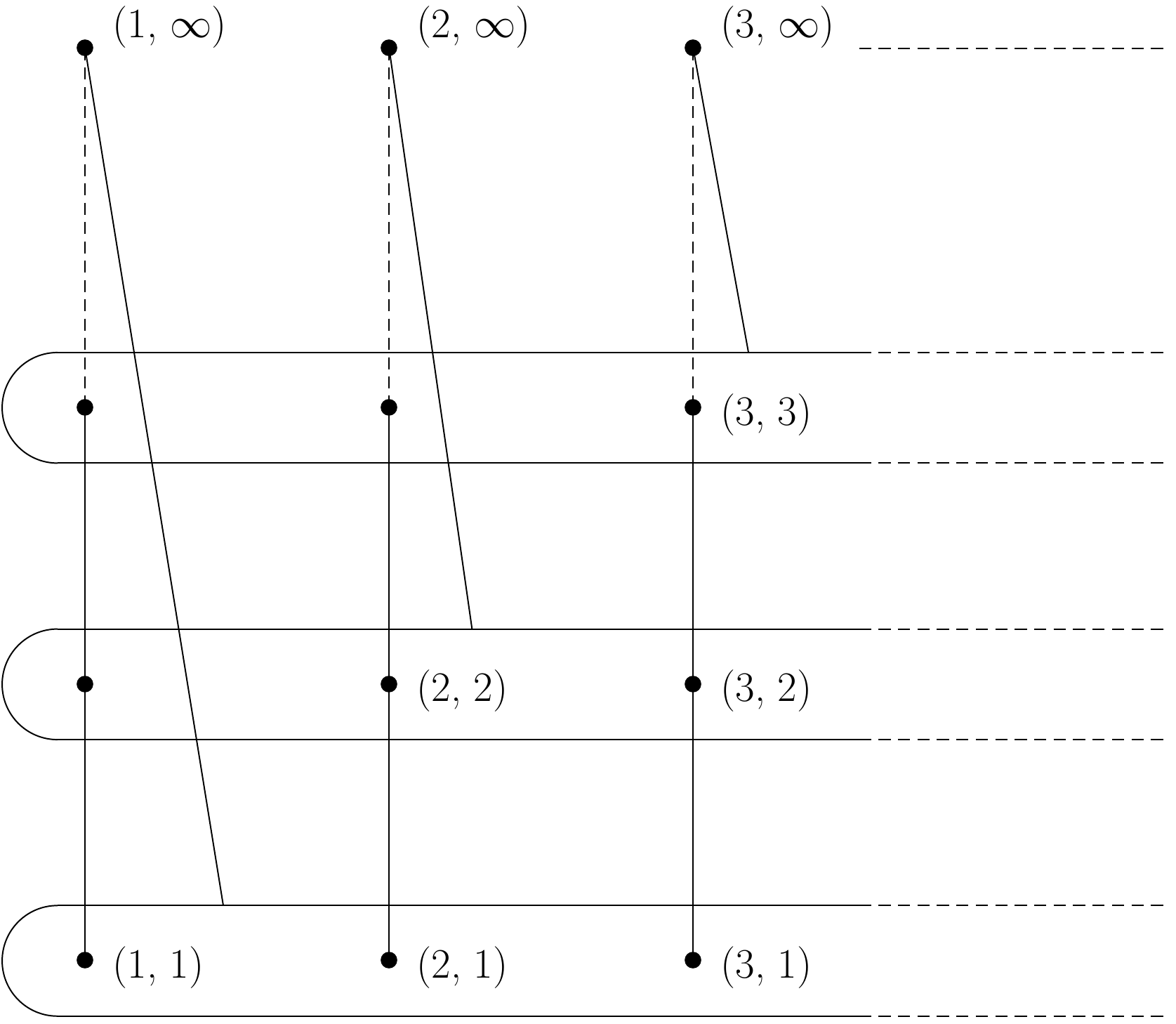}
\end{center}
 \caption{\label{fig:1} Johnstone's non-sober dcpo $\J$.}
\end{figure}

We will use the following convention throughout this subsection  
when we reason about $\J$. 

\begin{itemize}
\item $M$ denotes the set of all maximal points of $\J$, that is $M = \set{(i, \infty)}{i\in \mathbb N}$;
\item $M_{k} =  \set{(i, \infty)}{k< i}$, and $M_{k, l} =  \set{(i, \infty)}{k< i\leq l}$ for $k, l\in \N$ and $k< l$;
\item $N$ denotes the set $\J\setminus M$; elements in $N$ are of \emph{finite height};
\item $L_{i}$ denotes the set of points of $\J$ which are at \emph{Level} $i$, for each $i\in \N$, that is $L_{i} = \set{(j, i)}{j\in \mathbb N}$;
\item $C_{i}$ denotes the set of points of $\J$ which are in \emph{Column} $i$, that is $C_{i} = \set{(i, j)}{j\in \mathbb N\cup \oneset{\infty}}$. 
\item $D_{i}$ denotes the set $\bigcup_{j\leq i} C_{j}$, i.e., $D_{i}$ consists of elements in the first $i$-many columns.
\end{itemize}

The dcpo $\J$ in the Scott topology is a non-sober topological space, and the 
set $M$ of maximal points of $\J$ equipped with the relative
Scott topology is homeomorphic to $\Ncof$.

Let $\nu$ be an arbitrary\emph{bounded} continuous valuation on $\J$. 
We are going to show that $\nu$ can be written as a sum of a 
discrete valuation $\theta$ and $r\mu$, where $r$ is a nonnegative 
real number and in this section, $\mu$ is reserved for the fixed valuation on $\J$ that takes 
value $1$ on nonempty Scott-opens, and $0$ on the empty set:
$$ \mu (U) = \begin{cases}
1, & U \not= \emptyset;\\
0, & U = \emptyset.
\end{cases}$$
Note that $\mu$ on $\J$ is an analogue of $\beta$ on $\Ncof$. 
Indeed, the pushforward image of $\beta$ along the canonical 
topological embedding $n\mapsto (n, \infty) $ of $\Ncof$ into 
$\J$ is exactly the valuation $\mu$. 
Since $\theta$ and $r\mu$ are point-continuous (direct verification), by proving that $\nu$
is a sum of some discrete $\theta$
and $r\mu$, we infer that all (not necessarily bounded) continuous 
valuations on $\J$ are point-continuous by using a trick due to Heckmann. 

\begin{theorem}\label{main-theorem}
Every bounded continuous valuation $\nu$ on $\J$ is point-continuous. 
Moreover, there exist a discrete valuation 
$\theta$ and nonnegative real number $r$ such that $\nu = \theta +r\mu$. 
\end{theorem}

We prove this theorem by a series of results.  First, we give a lemma that will be
used a few times. It is a slight generalisation of \cite[Proposition 3.2]{heckmann96}.
\begin{lemma}
  \label{lemma:mu+nu}
  Let $\mu$ be a bounded continuous valuation on a space $X$, and
  $\nu$ be a monotonic map from $\O X$ to
  $\real \cup \{-\infty, +\infty\}$ such that $\nu (\emptyset)=0$.  If
  $\mu+\nu$ is a continuous valuation, then so is $\nu$.
\end{lemma}
\begin{proof}
  Since $\nu$ is monotonic and $\nu (\emptyset)=0$, $\nu$ actually
  takes its values in $\creal$.  It is clear that
  $\nu = (\mu+\nu) - \mu$ is modular, using the fact that $\mu$ is
  bounded for the subtraction to make sense.  The only challenge is
  Scott-continuity.  Let ${(U_i)}_{i \in I}$ be any directed family of
  open subsets of $X$, and $U$ be its union.  We have
  $\nu (U) \geq \dsup_{i \in I} \nu (U_i)$ by monotonicity.  In order
  to prove the reverse inequality, we consider any $a < \nu (U)$, and
  we show that $a \leq \nu (U_i)$ for some $i \in I$.  Since $\mu$ is
  bounded, $\mu (U) + a < (\mu+\nu) (U)$, and since $\mu+\nu$ is a
  continuous valuation, there is an $i \in I$ such that
  $\mu (U) + a < (\mu+\nu) (U_i)$.  Then
  $a < \mu (U_i) - \mu (U) + \nu (U_i) \leq \nu (U_i)$.
\end{proof}

\begin{lemma}\label{j-point-is-crescent}
For each element $a\in N$ (using the convention after Definition~\ref{def:johnstone}) , the singleton $\oneset{a}$ is a crescent. 
Thus, for each $a\in N$,  $\oneset{a}$ in the ring of sets 
generated by Scott-opens of $\J$.
\end{lemma}
\begin{proof}
For each $a\in N$, if $a$ is at Level $n$, that is $a = (j, n)$ for some $j\in \mathbb N$,
then $\oneset{a}$ can be written as $(\oneset{a}\cup \ua L_{n+1} )\setminus \ua L_{n+1}$. 
The proof is done since both $\oneset{a}\cup \ua L_{n+1}$ 
and $\ua L_{n+1}$ are Scott-open in $\J$. 
\end{proof}

\begin{proposition}
\label{minus}
For each Scott-open subset $U$ of $\J$, let 
$\nu^{*} (U)= \nu(U) - \sum_{a\in N} \nu_{\{a\}} (U)$. 
Then $\nu^*$ is a bounded continuous valuation on $\J$. 
\end{proposition}
\begin{proof}

Since for each $a\in N$, $\{a\}$ is a crescent by 
Lemma~\ref{j-point-is-crescent}, $\nu_{\{a\}}$ is a 
continuous valuation. Hence $\nu_{\{a\}}(U)$ makes 
sense for each $a\in N$. 

By Lemma~\ref{lemma:mu+nu}, we only need to prove that $\nu^{*}$ is well-defined and 
order-preserving.
Note that $N$ is a countable set. We index elements in $N$ by 
natural numbers by letting $N=\oneset{a_{1}, a_{2}, ..., a_{n},... }$. 
Since for each open set $U$, $\sum_{i=1}^{n} \nu_{\{a_{i}\}}(U) = 
\nu_{\{a_{1},..., a_{n}\}}(U)\leq \nu_{X}(U)=\nu(U)$ and $\nu(U)$ 
is bounded, it means that for each $U$ the sequence 
$\nu_{\{a_{i}\}}(U), n=1, ..., n,...$ is commutatively summable. 
Hence $\nu^{*}$ is well-defined and takes values in $\mathbb R_{+}$.

For monotonicity of $\nu^{*}$, we let $U$ be Scott-open and compute as follows:
\begin{align*}
\nu^{*}(U)   &= \nu(U) - \sum_{a\in N} \nu_{\{a\}} (U)\\
		 & = \nu(U) - \sum_{i= 1}^{\infty} \nu_{\{a_{i}\}}(U)\\
		 & = \lim_{n\to \infty} (\nu(U)-  \sum_{i= 1}^{n} \nu_{\{a_{i}\}}(U) )\\
		 & = \inf_{n\in \N} ( \nu_{X}(U)  -  \nu_{\{a_{1},..., a_{n}\}}(U)  )\\
		 &=\inf_{n\in \N} \nu_{X\setminus \{a_{1},..., a_{n}\}}(U) .
\end{align*}
Since for each $n\in \N$, $ \nu_{X\setminus \{a_{1},..., a_{n}\}}$ 
is a continuous valuation therefore a nonnegative order-preserving map. Hence the pointwise 
infimum $\nu^{*}$ of $\nu_{X\setminus \{a_{1},..., a_{n} \}}, n\in \N$ 
also is order-preserving. 
\end{proof}

\begin{lemma}\label{equal}
For any two Scott-open subsets $U, V$ of $\J$ with 
$M\cap U = M\cap V$, $\nu^{*}(U) = \nu^{*}(V)$.
\end{lemma}
\begin{proof}

Without loss of generality, we assume that $U\subseteq V$. 
Let $V_{n}= U \cup (V\cap  D_{n})$ (see the convention after
Definition~\ref{def:johnstone}). 
Since $M\cap U = M\cap V$, $V_{n}$ is Scott-open for 
each $n$. Moreover, $V = \bigcup_{n=1}^{\infty} V_{n}$.\\
Note that for every $n$, $V_{n}\setminus U = (V\cap D_{n})\setminus U$ is a finite 
subset of $N$, again by the fact that $M\cap U = M\cap V$. Then we know that 
\begin{align*}
\nu^{*}(V_{n})
& = \nu^{*}(U) + \nu^{*} (V_{n}) - \nu^{*}(U) \\
& = \nu^{*}(U) + (\nu(V_{n}) - \nu(U)) - (\sum_{a\in N} \nu_{\{a\}} (V_{n}) - \sum_{a\in N} \nu_{\{a\}} (U))\\
& = \nu^{*}(U) + \nu_{V_{n}\setminus U }(V_{n} )- (\sum_{a\in N} \nu_{\{a\}} (V_{n}) - \sum_{a\in N} \nu_{\{a\}} (U))\\
& = \nu^{*}(U) + \sum_{a\in V_{n}\setminus U}\nu_{\{a\}} (V_{n}) - \sum_{a\in N\cap ( V_{n}\setminus U )} \nu_{\{a\}}(V_{n})\\
& = \nu^{*}(U).
\end{align*}
Hence by Scott-continuity of $\nu^{*}$ (Proposition~\ref{minus}), we know that 
$ \nu^{*}(V) = \nu^{*}(\bigcup_{n\in \N} V_{n}) = \sup_{n\in \N} \nu^{*}(V_{n}) =  \nu^{*}(U)$.
\end{proof}

We consider $M$ as a subspace of $\J$ with the Scott-topology, 
and define a continuous valuation $\nu^{\infty}$ on $M$ by 
stipulating for each open subset $W$ of $M$ that $\nu^{\infty}(W) = \nu^{*}(U_{W})$, 
where $U_{W}$ is the largest Scott-open subset of $\J$ with $M \cap U_{W}= W$. 
We now use Lemma~\ref{equal} to prove that $\nu^{\infty}$ is indeed a bounded continuous valuation 
on $M$. This is a consequence of \cite[Proposition~5.2]{goubault21}, we provide a direct proof here nevertheless. 

\begin{lemma}
The map $\nu^{\infty}$ is a bounded continuous valuation on $M$ equipped with the relative Scott topology from $\J$. 
\end{lemma}
\begin{proof}
It is easy to see that $\nu^{\infty}(\emptyset) = 0$. 

For modularity, let $W_{1}$ and $W_{2}$ be two open subsets of $M$. Then we have:
\begin{align*}
\nu^{\infty}(W_{1}) + \nu^{\infty}(W_{2}) &=  \nu^{*}(U_{W_{1}}) + \nu^{*}(U_{W_{2}})   & \text{definition of $\nu^{\infty}$} \\
							   &= \nu^{*}(U_{W_{1}} \cup U_{W_{2}}) + \nu^{*}(U_{W_{1}} \cap U_{W_{2}}) &\text{$\nu^{*}$ is a valuation} \\
							   &= \nu^{*}(U_{W_{1}\cup W_{2}}) + \nu^{*}(U_{ W_{1} \cap W_{2} }) & \text{by Lemma~\ref{equal}}\\
							   &= \nu^{\infty}(W_{1} \cup W_{2}) + \nu^{\infty}(W_{1}\cap W_{2}).    & \text{definition of $\nu^{\infty}$}
\end{align*}
Now we prove that $\nu^{\infty}$ is Scott-continuous. Let $W_{i}, i\in I$ be a directed family of open subsets of $M$ and $W= \bigcup_{i\in I}W_{i} $. We first note that $U_{W_{i}}, i\in I$ also form a directed family of Scott-opens in $\J$. Hence we have 
\begin{align*}
\nu^{\infty}(W ) 		         &= \nu^{*}(U_{\bigcup_{i\in I}W_{i}})  & \text{definition of $\nu^{\infty}$}\\
		         &= \nu^{*}(\bigcup_{i\in I}U_{W_{i}})  & \text{by Lemma~\ref{equal}}\\
		         &= \sup_{i\in I} \nu^{*}(U_{W_{i}})  & \text{$\nu^{*}$ is Scott-continuous}\\
		         &= \sup_{i\in I} \nu^{\infty}(W_{i}).  & \text{definition of $\nu^{\infty}$}
\end{align*}
Finally, boundedness of $\nu^{\infty}$ is clear since $\nu^{*}$ is.
\end{proof}

We are now ready to prove Theorem~\ref{main-theorem}. 

\begin{proof}[{\bf Proof of Theorem~\ref{main-theorem}}]

The space $M$ in the relative Scott topology from $\J$ is 
homeomorphic to $\Ncof$. For example, one of the (infinitely many) 
homeomorphisms between them is the map sending $n\in \Ncof$ to $(n, \infty) \in M$. 
It then follows from Proposition~\ref{valuations-on-N} that $\nu^{\infty}$ 
can be written as $ \alpha + r \beta$ where $\alpha$ is a discrete valuation on 
$M$, $r$ is a nonnegative real number, and $\beta$ is the valuation on $M$ 
defined as $\beta(W) =1$ if $W \subseteq M$ is nonempty, and $\beta(W) = 0$ 
if $W = \emptyset$. For each Scott open subset $U$ of $\J$,  we have 
$$\nu^{*}(U) = \nu^{\infty}(U \cap M) = \alpha (U\cap M) + r \beta(U\cap M).$$
Now we define $\theta' (U) = \alpha(U\cap M)$. Since $\alpha$ is a discrete valuation on $M$, it is obvious that $\theta'$ is 
a discrete valuation on $\J$. Hence, by combining Proposition~\ref{minus} 
and the fact that $\mu(U) = \beta(U\cap M)$, we know that for each Scott 
open subset $U$ of $\J$,  
$$\nu(U) = \sum_{a\in N} \nu_{\{a\}} (U) + \nu^{*}(U) = \sum_{a\in N} \nu_{\{a\}} (U) + \theta'(U) +r\mu(U).$$ 
Finally, we let $\theta = \theta' + \sum_{a\in N} \nu_{\{a\}}$ 
and the proof is complete. 
\end{proof}

As a corollary to Theorem~\ref{main-theorem}, we have the following:

\begin{corollary}
Every continuous valuation on $\J$ is point-continuous. 
\end{corollary}
\begin{proof}
It is proved in \cite[Theorem 4.2]{heckmann96} that every continuous 
valuation on a $T_{0}$ topological space can be written as a directed 
supremum of bounded continuous valuations. By Theorem~\ref{main-theorem}, we know that on $\J$ 
every bounded continuous valuation is of the form $\theta + r\mu$, which is obviously point-continuous. Hence every 
continuous valuation on $\J$, as a directed supremum of point-continuous 
valuations, also is point-continuous~\cite[Section 3.2, Item~(5)]{heckmann96}.
\end{proof}

\subsection{Non-minimality of $\mu$ on $\J$}

In this subsection, we show that the valuation $\mu$ on $\J$, defined in the last
subsection, is not a minimal valuation. 

As we know from Theorem~\ref{main-theorem}, every bounded continuous 
 valuation $\nu$ on $\J$ is of the form $\theta + r\mu$. 
We first prove that that for any $r>0$, $\theta + r\mu$ can not be 
written as a supremum of discrete valuations. Without loss of 
generality, in the sequel we assume that the total mass of $\theta + r\mu$ is $1$, i.e., 
$\theta (\J) + r\mu(\J) = \theta (\J)+ r = 1$. 

We will need the following extension result for discrete 
valuations, which is a special case of \cite[Theorem 4.1]{alvarez99}. 
\begin{lemma}\label{discrete-extension}
Let $\theta$ be a discrete valuation on a dcpo $D$ with 
$\theta(D) < \infty$. Then $\theta$ has a unique extension to a 
measure, which we again denote as ${\theta}$, on the Borel 
$\sigma$-algebra of $(D, \sigma D)$, where $\sigma D$ is the Scott topology on $D$.
In fact, if $\theta = \sum_{i\in \N} r_{i}\delta_{x_{i}}$ with 
$\sum_{i\in \N}r_{i} < \infty$,  then the value of its measure 
extension on a Borel subset $B$ of $D$ is just $\sum_{i/x_{i}\in B}r_{i}$. \hfill $\Box$
\end{lemma}

Lemma~\ref{j-point-is-crescent} says that for each $a\in N\subseteq \J$, 
the singleton $\oneset{a}$ is a crescent, hence $\{a\}$ is a Borel subset 
of $\J$. Actually, every subset of $\J$ is a Borel subset. To see this, we only 
need to show that $\oneset{(i,\infty)}$ is a Borel subset in $\J$ for each $i\in \N$, 
since $\J$ is countable. Indeed, 
$\oneset{(i,\infty)} = \da (i,\infty) \setminus \bigcup_{i\in \N} \da L_{j}$, 
and for $i, j\in \N$, $\da (i,\infty)$ 
and $\da L_{j}$ are Scott-closed. Now if $\theta$ is a discrete 
valuation on $\J$, then it makes sense to apply $\theta$ to any subset of $\J$, 
viewing that $\theta$ is a measure defined on the Borel $\sigma$-algebra 
generated by Scott-opens on $\J$.

\begin{proposition}\label{cannot-be-a-sup}
Let $\theta$ be a discrete valuation on $\J$ with total mass $1-r$ and $r>0$.  
Then $\theta+r\mu$ is not a directed supremum of discrete valuations. 
\end{proposition}
\begin{proof}
Suppose, for the sake of contradiction,  that there exists a directed family 
$\oneset{\theta_{a}}_{a\in A}$ of discrete valuations with supremum 
$\sup_{a\in A}\theta_{a} = \theta + r\mu$. 

Since $(\theta+r\mu)(\J) =1$, there exists a discrete valuation 
$\theta_{a}, a\in A$ such that $\theta_{a}(\J) > 1- \frac{r}{4}$. 
Note that $\sup_{i\in \N}D_{i} = \J$, where $D_{i} = \bigcup_{j\leq i}C_{j}$ 
consists of points in the first $i$-many columns (see Section~\ref{section-non-soberj}), 
and also that $D_{i}$ is a Borel subset of $\J$ for each $i\in \N$. Since $\theta_{a}$ 
extends to a measure (Lemma~\ref{discrete-extension}), we know that there 
exists a $k\in \N$ such that $\theta_{a}(D_{k})> 1-\frac{r}{4}$. 

Obviously, $\da D_{k}$ is Scott-closed and $U_{k}= \J\setminus \da D_{k}$ 
is a nonempty Scott-open subset of $\J$, hence
$(\theta + r\mu)(U_{k}) \geq r > \frac{3r}{4}$. Again since 
$\theta+r\mu = \sup_{a\in A} \theta_{a} $, there exists 
$\theta_{b} > \theta_{a}$ such that $\theta_{b}(U_{k})> \frac{3r}{4}$. 

Note  that the set $\ua D_{k}$ is a filtered intersection of countably many
Scott-open subsets of $\J$, for example one can write 
$\ua D_{k} = \bigcap_{i\in \N} O_{i}$, where for each $i$, 
$O_{i} = \J \setminus \da \{(k+1, i), (k+2, i), ...  \}$. Then we know 
that $$\theta_{b}(\ua D_{k}) = \inf_{i\in \N} \theta_{b}(O_{i}) 
\geq \inf_{i\in \N} \theta_{a}(O_{i}) = \theta_{a}(\ua D_{k}) 
\geq \theta_{a} (D_{k}) > 1- \frac{r}{4}.$$

Since $M_{k} =  \set{(i, \infty)}{k< i}$ is a Borel subset of $\J$ and it 
equals to $\ua D_{k} \cap U_{k}$, hence by inclusion-exclusion of $\theta_{b}$ we know 
$$\theta_{b}(M_{k}) + \theta_{b}(\J) \geq   \theta_{b}(U_{k})  + \theta_{b}(\ua D_{k}) > 1+\frac{r}{2},$$ 
from which it follows that $\theta_{b}(M_{k}) > \frac{r}{2}$. 
Since $\bigcup_{l/l > k} M_{k, l} = M_{k}$ ($M_{k, l} =\set{(i, \infty)}{k< i \leq l} $), there exists a big 
enough $l> k$ such that $\theta_{b}(M_{k, l}) >  \frac{r}{2}$. 

Now consider $U_{l}= \J\setminus \da D_{l}$. Since 
$(\theta+r\mu)(U_{l}) \geq r > \frac{3r}{4}$, we find a $\theta_{c}, c\in A$ 
such that $\theta_{c}\geq \theta_{b}$ and $\theta_{c}(U_{l})> \frac{3r}{4}$. 
Meanwhile, note that $D_{k} \subseteq D_{l}$, similar to the reasoning above 
we will have that 
$\theta_{c}(\ua D_{l}) \geq \theta_{b}(\ua D_{l}) \geq \theta_{b}(\ua D_{k})> 1-\frac{r}{4}$ and that $\theta_{c}(M_{l})> \frac{r}{2}$. 
Then there exists a large enough natural number $m> l$ such that 
$\theta_{c} (M_{l,m}) > \frac{r}{2}$. 

We claim that $\theta_{c}(\J) > r$. Indeed, 
$\theta_{c}(\J) \geq \theta_{c} (M_{k,m}) = \theta_{c} (M_{k,l}) + \theta_{c}(M_{l,m}) \geq \theta_{b}(M_{k, l}) + \theta_{c}(M_{l,m}) > r$. 
The second to last inequality comes from the fact that $\theta_{c}\geq \theta_{b}$ and 
$M_{k, l}$ is a filtered intersection of countably many Scott-open subsets of $\J$: 
one may take such a filtered family as $\set{ \J\setminus \da F }{ F\subseteq_{\mathrm{fin}} M\setminus M_{k, l}}$.

Next we consider the Scott-open set $U_{m} = \J\setminus \da D_{m}$ and then find $\theta_{d}\geq \theta_{c}$ with $\theta_{d}(\J) > \frac{3r}{2}$, and so forth. We proceed the above process $N$ times, where $N$ is a natural number 
satisfying that $N\times \frac{r}{2} > 1$, to find a discrete valuation 
$\theta_{z}$ in the directed family $\oneset{\theta_{a}}_{a\in A}$. 
By the construction, we would know that $\theta_{z}(\J) > 1$. However, this is 
impossible since $\theta_{z} \leq \theta+r\mu$ and $(\theta+r\mu)(\J) = 1$. 
So our assumption of the existence of the directed family  
$\oneset{\theta_{a}}_{a\in A}$ of discrete valuations with 
$\sup_{a\in A}\theta_{a} = \theta + r\mu$ mush have been wrong, 
and we finish the proof. 
\end{proof}

\begin{theorem}
Let $\theta$ be a discrete valuation on $\J$ with total mass $1-r$ and $r>0$.  
Then $\theta+r\mu$ is not a minimal valuation. 
\end{theorem}
\begin{proof}
We know that the set of minimal valuations (bounded by $1$) on $\J$ is obtained 
by taking directed suprema of simple valuations (bounded by $1$), directed suprema 
of directed suprema of simple valuations (bounded by $1$), and so forth, transfinitely.  
However, from Theorem~\ref{main-theorem} and Proposition~\ref{cannot-be-a-sup}, 
at each of these steps we only obtain discrete valuations.  Hence, by transfinite 
induction $\theta+r\mu$ is not a minimal valuation for $r>0$.
\end{proof}

The following result, as promised, is then obvious:

\begin{corollary}\label{coro:muisnotpointcontinuous}
The valuation $\mu$ on $\J$ is point-continuous but not minimal. 
\end{corollary}

\section{Continuous valuations need not be point-continuous}

In the last section we have seen that \emph{on dcpo's} not all point-continuous valuations 
are minimal. In this section, we present another separation result: on dcpo's, not all continuous
valuations are point-continuous,  a result which is known on general topological spaces, but unknown
on dcpo's with the Scott topology. 

\subsection{What are non-point-continuous valuations like?}
\label{sec:any-non-point}

Tix showed that, on a sober space $X$, the bounded continuous
valuations that take only finitely many values are the simple
valuations, namely the finite linear combinations
$\sum_{i=1}^n a_i \delta_{x_i}$, where each coefficient $a_i$ is in
$\Rp$ and $x_i \in X$ \cite[Satz~2.2]{tix95}.  Hence they are all
point-continuous.

We will extend that result slightly in this subsection. We first recall that
a nonempty subset $A$ of a space $X$ is \emph{irreducible} if $A \subseteq B\cup C$,
for closed subsets $B$ and $C$ of $X$, implies that $A\subseteq B$ or $A\subseteq C$. $X$ 
is \emph{sober} if every irreducible closed subset $C$ of $X$ is the closure of some 
unique singleton subset of $X$.

\begin{lemma}
  \label{lemma:egame}
  Let $X$ be a topological space.  For every irreducible closed subset
  $C$ of $X$, let $\egame C \colon \O X \to \creal$ map every open
  subset $U$ of $X$ to $1$ if $U$ intersects $C$, to $0$ otherwise.
  Then $\egame C$ is a point-continuous valuation.  If $C = \dc x$ for
  some point $x$, then $\egame C = \delta_x$.
\end{lemma}
\begin{proof}
  It is clear that $\egame C$ is strict.  In order to show
  point-continuity, we assume $0 \leq r < \egame C (U)$.  We note that
  $C$ must intersect $U$, say at $x$, and that $0 \leq r < 1$.  Let
  $A \eqdef \{x\}$.  For every open neighborhood $V$ of $A$, $V$
  intersects $C$ (at $x$), so $\egame C (V)=1 > r$.

  In order to show modularity, we observe that, since $C$ is
  irreducible, for all open subsets $U$ and $V$ of $X$, $C$ intersects
  both $U$ and $V$ if and only if $C$ intersects $U \cap V$.  Hence
  $\egame C (U \cap V) = \min (\egame C (U), \egame C (V))$.  Since
  $C$ intersects $U$ or $V$ if and only if $C$ intersects $U \cup V$,
  we have $\egame C (U \cup V) = \max (\egame C (U), \egame C (V))$.
  Now for all real numbers $a$ and $b$, $\max (a,b)+\min (a,b) = a+b$,
  so
  $\egame C (U \cap V) + \egame C (U \cup V) = \egame C (U) + \egame C
  (V)$.

  When $C = \dc x$, for every $U \in \O X$ we have $\egame C (U)=1$ if
  and only if $C$ intersects $U$, if and only if $x \in U$, if and
  only if $\delta_x (U)=1$.
\end{proof}

For every continuous valuation $\nu$ on a space $X$, let $\Val (\nu)$
denote the set of non-trivial values $\{\nu (U) \mid U \in \O X, \nu
(U) \neq 0, +\infty\}$ taken by $\nu$.
\begin{lemma}
  \label{lemma:nu:min}
  Let $\nu$ be a continuous valuation on a space $X$, with the
  property that $\Val (\nu)$ has a least element $r$.  Then there is an
  irreducible closed subset $C$ of $X$ such that:
  \begin{enumerate}
  \item[(1)] $\nu' \eqdef \nu - r \egame C$ is a continuous valuation;
  \item[(2)] $\Val (\nu') \subseteq \{v-r \mid v \in \Val (\nu), v \neq r\}$.
  \end{enumerate}
\end{lemma}
\begin{proof}
  Let $U_r$ be any open subset of $X$ such that
  $\nu (U_r) = r \neq 0$.  We consider the family $\mathcal U$ of open
  subsets $U$ of $X$ such that $\nu (U \cap U_r)=0$, or equivalently
  $\nu (U \cap U_r) < r$.  That family contains the empty set, and for
  any two elements $U$, $V$ of $\mathcal U$, we have
  $\nu ((U \cup V) \cap U_r) \leq \nu (U \cap U_r) + \nu (V \cap U_r)
  = 0$, so $U \cup V$ is in $\mathcal U$.  It follows that
  $\mathcal U$ is directed.  Let $U_*$ be the union of all the
  elements of $\mathcal U$.  Then $\nu (U_* \cap U_r) = \dsup_{U
    \in \mathcal U} \nu (U \cap U_r) = 0$, so $U_*$ is in
  $\mathcal U$.  It follows that $U_*$ is the largest element of
  $\mathcal U$.  In particular, for every $U \in \O X$, $\nu (U \cap
  U_r) = 0$ if and only if $U \subseteq U_*$.

  We define $C$ as the complement of $U_*$.  Since
  $\nu (U_r) \neq 0$, $U_r$ is not included in $U_*$, so $C$
  intersects $U_r$.  In particular, $C$ is non-empty.  For any two
  open subsets $U$ and $V$ of $X$ that intersect $C$, we claim
  that $U \cap V$ also intersects $C$.  This will show that $C$ is
  irreducible.  By assumption, neither $U$ nor $V$ is included in
  $U_*$, so $\nu (U \cap U_r) \geq r$ and
  $\nu (V \cap U_r) \geq r$.  It follows that
  $\nu ((U \cap V) \cap U_r) = \nu (U \cap U_r) + \nu (V \cap U_r) -
  \nu ((U \cup V) \cap U_r) \geq 2r-r = r$, using modularity and the
  inequality $\nu ((U \cup V) \cap U_r) \leq \nu (U_r) = r$.  We
  conclude that $U \cap V$ is not included in $U_*$, hence
  intersects $C$.

  $(1)$ Let $\nu' \eqdef \nu - r \egame C$.  We first verify that $\nu'$
  is monotonic.  Let $U \subseteq V$ be two open subsets of $X$.  If
  $\egame C (U) = \egame C (V)$, then $\nu' (U) \leq \nu' (V)$.
  Therefore let us assume that $\egame C (U) \neq \egame C (V)$, hence
  necessarily $\egame C (U)=0$, $\egame C (V)=1$.  As a consequence,
  $\nu (U \cap U_r)=0$ and $\nu (V \cap U_r) \geq r$.  Then:
  \begin{align*}
    \nu' (U) = \nu (U) & = \nu (U \cup U_r) + \nu (U \cap U_r) - \nu
                         (U_r) \\
                       & \text{by modularity, and since }\nu (U_r) < +\infty \\
                       & = \nu (U \cup U_r) - r \\
                       & \leq \nu (V \cup U_r) - r \\
                       & = \nu (V) + \nu (U_r) - \nu (V \cap U_r) - r \\
                       & \text{by modularity, and since }\nu (V \cap U_r) \leq \nu (U_r)
                         < +\infty \\
                       & \leq \nu (V) + \nu (U_r) - r - r = \nu' (V).
  \end{align*}
  In particular, since $\nu' (\emptyset) = 0$, $\nu'$ takes its values
  in $\creal$. It then follows from Lemma~\ref{lemma:mu+nu} that
  $\nu'$ is a continuous valuation.

  $(2)$  Let $V' \eqdef \{v-r \mid v \in \Val (\nu), v\neq r\}$.  For
  every $U \in \O X$, if $U$ intersects $C$ then
  $\nu' (U) = \nu (U) - r$; hence if $\nu' (U) \neq 0, +\infty$, then
  $\nu' (U) \in V'$.  Otherwise, $U \subseteq U_*$, so
  $\nu (U \cap U_r) = 0$, and therefore
  $\nu' (U) = \nu (U) = \nu (U \cup U_r) + \nu (U \cap U_r) - \nu
  (U_r) = \nu (U \cup U_r) - r$, using modularity and the fact that
  $\nu (U_r) = r < +\infty$.  Hence if $\nu' (U) \neq 0, +\infty$,
  then $\nu' (U)$ is in $V'$.
\end{proof}

\begin{proposition}
  \label{prop:tix}
  Let $X$ be a topological space.  The bounded continuous valuations
  $\nu$ on $X$ that take only finitely many values are exactly the
  finite linear combinations $\sum_{i=1}^n a_i \egame {C_i}$, where
  each $C_i$ is irreducible closed and $a_i \in \Rp \diff \{0\}$.
\end{proposition}
\begin{proof}
  That $\sum_{i=1}^n a_i \egame {C_i}$ only takes finitely many values
  is obvious.  We prove the converse implication by induction on the
  number $n$ of non-zero values taken by $\nu$.  If $n=0$, then $\nu$
  is the zero valuation.  Otherwise, let $r$ be the least non-zero
  value taken by $\nu$.  We find $C$ and $\nu'$ as in
  Lemma~\ref{lemma:nu:min}.  Since $\Val (\nu')$ has one less element
  than $\Val (\nu)$, we can apply the induction hypothesis, allowing us
  to conclude.
\end{proof}

\begin{proposition}
  \label{prop:fin:pcont}
  Every continuous valuation $\nu$ on a topological space $X$ that
  takes only finitely many distinct values is point-continuous.
\end{proposition}
\begin{proof}
  If $\nu$ is bounded, then by Proposition~\ref{prop:tix}, $\nu$ is of
  the form $\sum_{i=1} a_i \egame {C_i}$, where each $C_i$ is
  irreducible closed and $a_i \in \Rp$, and $\egame {C_i}$ is
  point-continuous.  Any linear combination of point-continuous
  valuations is point-continuous \cite[Section~3.2]{heckmann95}, so
  $\nu$ is point-continuous.

  Hence we concentrate on the case where $\nu (X)$ is equal
  to $+\infty$.  Let $s$ be the greatest finite value that $\nu$ takes.
  Consider the family $\mathcal S$ of all opens $U$ such that
  $\nu(U) \leq s$. For any $U, V\in \mathcal S$, $\nu(U\cup V) =
  \nu(U) + \nu(V) - \nu(U\cap V) $, which is obviously a 
  finite value. Hence $U\cup V$ is in $S$. This implies that the 
  family $S$ is directed. Let $U_{s}$ be the directed union of $\mathcal S$,
  and $C_{\infty}$ be the complement of $U_{s}$.  We notice 
  that $\nu(U_{s}) = \sup_{U\in \mathcal S}\nu(U) = s$, and hence
  that $U_{s}$ is a proper subset of $X$.  
  
  The collection $\mathcal F$ of simple
  valuations $\sum_{i=1}^n a_i \delta_{x_i}$ such that every $x_i$ is
  in $C_\infty$ is directed: it is non-empty because $C_\infty$ is
  non-empty, and any two elements $\sum_{i=1}^n a_i \delta_{x_i}$ and
  $\sum_{i=1}^n b_i \delta_{x_i}$ (which we take over the same set of
  points $x_i$, without loss of generality), have an upper bound, such
  as $\sum_{i=1}^n \max (a_i,b_i) \delta_{x_i}$.  Let $\mu$ be the
  supremum of $\mathcal F$.  Since the family of point-continuous
  valuations is closed under directed suprema
  \cite[Section~3.2~(5)]{heckmann95}, $\mu$ is point-continuous.

  We compute $\mu$ explicitly.  For every open subset $U$ of $X$, if
  $U \subseteq U_s$, namely if $U$ and $C_\infty$ are disjoint,
  then the value of any element $\sum_{i=1}^n a_i \delta_{x_i}$ of
  $\mathcal F$ on $U$ is zero, so $\mu (U)=0$.  Otherwise, let us pick
  an element $x$ from $U \cap C_\infty$.  Then $a \delta_x$ is in
  $\mathcal F$ for every $a \in \Rp$, so $\mu (U) \geq a \delta_x (U)
  = a$ for every $a \in \Rp$, from which it follows that $\mu (U) = +\infty$.

  Let $\nu_{|U_s}$ be the restriction of $\nu$ to $U_s$,
  namely the map $V \mapsto \nu (U_s \cap V)$.  This is a
  continuous valuation \cite[Section~3.3]{heckmann95}.  It is bounded
  by construction of $U_s$, and takes only finitely many values.
  Therefore, as we have already seen, Proposition~\ref{prop:tix}
  entails that $\nu_{|U_s}$ is point-continuous.

  We now observe that $\nu = \nu_{|U_s} + \mu$.  For every
  $U \in \O X$, if $U \subseteq U_s$, then
  $\nu_{|U_s} (U) + \mu (U) =  \nu (U \cap U_s) + 0 = \nu (U)$.
  Otherwise
  $\nu_{|U_s} (U) + \mu (U) = \nu_{|U_s} (U) + (+\infty) =
  +\infty = \nu (U)$.  Being a sum of two point-continuous valuations,
  $\nu$ is point-continuous.
\end{proof}

\subsection{The Sorgenfrey line}
\label{sec:sorgenfrey}

Let $\real$ be the set of real numbers, with its usual metric
topology.

The \emph{Sorgenfrey line} $\Rl$ has the same set of points as
$\real$, but its topology is generated by the half-open intervals
$[a, b[$, $a < b$ \cite{Sorgenfrey}.  The topology of $\Rl$ is finer
than that of $\real$.  $\Rl$ is a zero-dimensional, first-countable
space \cite[Exercices~4.1.34, 4.7.17]{goubault13a}.  It is paracompact
hence $T_4$, and Choquet-complete hence a Baire space
\cite[Exercises~6.3.32, 7.6.11]{goubault13a}.  $\Rl$ is not locally
compact, as every compact subset of $\Rl$ has empty interior
\cite[Exercise~4.8.5]{goubault13a}.  In fact, $\Rl$ is not even
consonant \cite{Bouziad:Borel,CW:dissonant}.  Although it is
first-countable, $\Rl$ is not second-countable
\cite[Exercise~6.3.10]{goubault13a}.

A \emph{hereditarily Lindel\"of space} is a space in which every
family ${(U_i)}_{i \in I}$ of open subsets has a countable subfamily
with the same union, or equivalently a space whose subspaces are all
Lindel\"of.  Every second-countable space is hereditarily Lindel\"of,
but $\Rl$ is a counterexample to the reverse implication, as the
following folklore result demonstrates.

\begin{proposition}
  \label{prop:Rl:hL}
  $\Rl$ is hereditarily Lindel\"of.
\end{proposition}
\begin{proof}
  Let us denote by $\mathring U$ the interior of any set $U$ in the
  topology of $\real$ (not $\Rl$).

  Let ${(U_i)}_{i \in I}$ be any family of open subsets of $\Rl$,
  $U \eqdef \bigcup_{i \in I} U_i$ and
  $U' \eqdef \bigcup_{i \in I} \mathring U_i$.  Since $\real$ is
  second-countable hence hereditarily Lindel\"of, there is a countable
  subset $J$ of $I$ such that $U' = \bigcup_{i \in J} \mathring U_j$.

  We claim that $U \diff U'$ is countable.  For each point
  $x \in U \diff U'$, $x$ is in some $U_i$ hence in some basic open
  set $[a, b[ \subseteq U_i$.  Then $x$ is also included in the
  smaller basic open set $[x, b[$.  Let us write $b$ as
  $x + \delta_x$, with $\delta_x > 0$.  We observe that
  $]x, x+\delta_x[$ is included in $\mathring U_i$, hence in $U'$.

  Let $x$ and $y$ be two points of $U \diff U'$, with $x < y$.  If
  $[x, x+\delta_x[$ and $[y, y+\delta_y[$ intersect, then $y$ is in
  $[x, x+\delta_x[$, and since $x \neq y$, it follows that $y$ is in
  $]x, x+\delta_x[$.  That is impossible, since $]x, x+\delta_x[$ is
  included in $U'$ and $y$ is not in $U'$.

  Therefore, for any two distinct points $x$ and $y$ of $U \diff U'$,
  $[x, x+\delta_x[$ and $[y, y+\delta_y[$ are disjoint.  We pick one
  rational number $q_x$ in each set $[x, x+\delta_x[$ with $x \in U
  \diff U'$: then $x \neq y$ implies $q_x \neq q_y$, and therefore $U
  \diff U'$ is countable.

  Let us now pick an index $i_x \in I$ such that $x \in U_{i_x}$, one
  for each $x \in U \diff U'$.  Then
  $U = U' \cup (U \diff U') \subseteq U' \cup \bigcup_{x \in U \diff
    U'} U_{i_x} \subseteq U$, so
  ${(U_i)}_{i \in J \cup \{i_x \mid x \in U \diff U'\}}$ is a
  countable subfamily with the same union as our original family
  ${(U_i)}_{i \in I}$.
\end{proof}

\begin{lemma}
  \label{lemma:Rl:Borel}
  $\real$ and $\Rl$ have the same Borel $\sigma$-algebra.
\end{lemma}
\begin{proof}
  Let $\Sigma$ be the Borel $\sigma$-algebra of $\real$, and
  $\Sigma_\ell$ be that of $\Rl$.  Clearly,
  $\Sigma \subseteq \Sigma_\ell$.

  In the converse direction, we first claim that every open subset $U$
  of $\Rl$ is in $\Sigma$.  For every $x \in U$, the union of all the
  intervals included in $U$ and containing $x$ is convex, hence is an
  interval $I_x$.  Let $a_x$ be its lower end, $b_x$ be its upper end.
  We note that $b_x$ cannot be in $I_x$, otherwise $b_x$ would be in
  some basic open subset $[b_x, b_x+\delta_x[$ included in $U$, which
  would allow us to form a strictly larger interval included in $U$
  and containing $x$.  Hence $I_x$ is equal to $[a_x, b_x[$ or to
  $]a_x, b_x[$.  In any case, $a_x < b_x$, so $I_x$ contains a
  rational number $q_x$.  Since the intervals $I_x$ are pairwise
  disjoint, there are no more such intervals than there are rational
  numbers.  Moreover, $\bigcup_{x \in U} I_x$ is equal to $U$.  It
  follows that $U$ is a countable union of pairwise disjoint
  intervals, and is therefore in $\Sigma$.

  Since $\Sigma_\ell$ is the smallest $\sigma$-algebra containing the
  open subsets of $\Rl$, we conclude that
  $\Sigma_\ell \subseteq \Sigma$.
\end{proof}

Given a topological space $X$, a Borel measure on $X$ is
\emph{$\tau$-smooth} if and only if its restriction to the lattice
$\O X$ of open subsets of $X$ is a continuous valuation.  Adamski
showed that a space $X$ is hereditarily Lindel\"of if and only if
every Borel measure on $X$ is $\tau$-smooth
\cite[Theorem~3.1]{Adamski:measures}.
\begin{proposition}
  \label{prop:hL:contval}
  For every Borel measure $\mu$ on $\real$, the restriction of $\mu$
  to the open subsets of $\Rl$ is a continuous valuation.
\end{proposition}
\begin{proof}
  By Lemma~\ref{lemma:Rl:Borel}, $\mu$ is also a Borel measure on
  $\Rl$.  We then apply Adamski's theorem, thanks to
  Proposition~\ref{prop:Rl:hL}.
\end{proof}

\begin{corollary}
  \label{corl:lambda}
  The restriction $\lambda$ of the Lebesgue measure on $\real$ to
  $\O \Rl$ is a continuous valuation.  \qed
\end{corollary}

\subsection{The compact subsets of $\Rl$}
\label{sec:compact-subsets-rl}

We recall that every compact subset $Q$ of $\Rl$ has empty interior.
For completeness, we give the proof here.  Let us assume a compact
subset $Q$ of $\Rl$ with non-empty interior.  $Q$ contains a basic
open set $[a, b[$ with $a < b$.  $[a, b[$ is not only open, but also
closed, since it is the complement of the open set
$]-\infty, a[ \cup [b, +\infty[$.  (Those two sets are open, being
equal to $\bigcup_{m \in \N} [a-m-1, a-m[$ and to
$\bigcup_{n\in \N} [b+n, b+n+1[$ respectively).  Being closed in a
compact set, $[a, b[$ is compact.  But the open cover
${([a, b-\epsilon[)}_{\epsilon \in ]0, b-a[}$ of $[a, b[$ has no
finite subcover: contradiction.

We use the following folklore result.  The fact that every compact
subset of $\Rl$ is countable is the only thing we will need to know in
later sections, together with Corollary~\ref{corl:Rl:seq} below, but we
think that giving a complete characterization of the compact subsets
of $\Rl$ is interesting in its own right, and may help one understand
better what they look like.
\begin{lemma}
  \label{lemma:Q:countable}
  Every compact subset $Q$ of $\Rl$ is countable, bounded, and is
  well-founded in the ordering $\geq$.
\end{lemma}
\begin{proof}
  For every point $x$ of $Q$, the family of open sets
  $]-\infty, x-\epsilon[$ ($\epsilon > 0$) plus $[x, +\infty[$ is an
  open cover of $Q$ (in fact, of the whole of $\Rl$), hence contains a
  finite subcover.  It follows that $Q$ is contained in
  $]-\infty, x-\epsilon_x[ \cup [x, +\infty[$ for some
  $\epsilon_x > 0$.  Equivalently, $Q$ contains no point in
  $[x-\epsilon_x, x[$.

  From this, we deduce that $[x-\epsilon_x, x[$ and
  $[y-\epsilon_y, y[$ are disjoint for any two distinct points $x$ and
  $y$ of $Q$.  Indeed, without loss of generality, let us assume that
  $y < x$.  Since $y$ is in $Q$, it is not in $[x-\epsilon_x, x[$, so
  $y > x$ (which is impossible), or $y \leq x-\epsilon_x$.  Then
  $[y-\epsilon_y, y[$ lies entirely to the left of
  $[x-\epsilon_x, x[$, and does not intersect it.

  Next, for each $x \in Q$, we pick a rational number $q_x$ in
  $[x-\epsilon_x, x[$.  By the disjointness property we just proved,
  the map $x \in Q \mapsto q_x$ is injective, so $Q$ is countable.

  Since $Q$ is compact in $\Rl$ and the topology of $\Rl$ is finer
  than that of $\R$, $Q$ is also compact in $\real$, hence is bounded.
  
  Let us imagine that $Q$ contains an infinite increasing sequence
  $r_0 < r_1 < \cdots < r_n < \cdots$.  Let also
  $r \eqdef \dsup_{n \in \N} r_n$.  Then the open sets
  $]-\infty, r_0[$, $[r_0, r_1[$, \ldots, $[r_n, r_{n+1}[$, \ldots,
  and $[r, +\infty[$ (if $r < +\infty$, otherwise we ignore the last
  interval) form an open cover of $Q$ without a finite subcover.  It
  follows that $Q$ is well-founded in the ordering $\geq$.
\end{proof}

We now give a few seemingly lesser known results.  Given a compact
subset $Q$ of $\Rl$, $Q$ is a chain, namely a totally ordered subset
of $(\real, \geq)$.  Notice that we use the reverse ordering $\geq$:
Lemma~\ref{lemma:Q:countable} tells us that $Q$ is even a well-founded
chain.  A chain $A$ in a poset $(P, \sqsubseteq)$ is a \emph{subdcpo}
of $P$ if and only if for every non-empty (equivalently, directed)
subset $D$ of $A$, the supremum of $D$ exists in $P$ and is in $A$.

\begin{lemma}
  \label{lemma:Q:limembed}
  Every compact subset $Q$ of $\Rl$ is a subdcpo of $(\real, \geq)$:
  for every non-empty subset $D$ of $Q$, $\inf D$ is in $Q$.  If $Q$
  is non-empty, then $Q$ has a least element in the usual ordering
  $\leq$.
\end{lemma}
\begin{proof}
  We first need to note that a net ${(x_i)}_{i \in I, \sqsubseteq}$
  converges to $x$ in $\Rl$ if and only if $x_i$ tends to $x$ from the
  right, namely: for every $\epsilon > 0$, $x \leq x_i < x+\epsilon$
  for $i$ large enough \cite[Exercise~4.7.6]{goubault13a}.

  Let $D$ be any non-empty subset of $Q$, and let $r \eqdef \inf D$.
  By Lemma~\ref{lemma:Q:countable}, $Q$ is bounded, so $r$ is a
  well-defined real number.  Since $Q$ is well-founded, $D$ is
  isomorphic to a unique ordinal $\beta$, and we can write the
  elements of $D$ as $r_\alpha$, $\alpha < \beta$, in such a way that
  for all $\alpha, \alpha' < \beta$, $\alpha \leq \alpha'$ if and only
  if $r_\alpha \geq r_{\alpha'}$.  The net ${(r_\alpha)}_{\alpha <
    \beta, \leq}$ then converges to $r$ from the right.   Since $\Rl$
  is $T_2$, its compact subset $Q$ is closed, so $r$ is in $Q$.

  If $Q$ is non-empty, we can take $D \eqdef Q$.  Then $\inf D \in Q$
  is the least element of $Q$.
\end{proof}

Another way of expressing Lemma~\ref{lemma:Q:limembed}, together with
the well-foundedness property of Lemma~\ref{lemma:Q:countable} is to
say the following.  The fact that $\beta_Q$ is not a limit ordinal is
due to the fact that $Q$ must be empty or have a least element, namely
that $\beta_Q$ must be equal to $0$ or have a largest element.
\begin{lemma}
  \label{lemma:Q:fQ:limit}
  For every compact subset $Q$ of $\Rl$, $(Q, \geq)$ is
  order-isomorphic to a unique ordinal $\beta_Q$; $\beta_Q$ is not a
  limit ordinal, and the order-isomorphism is a Scott-continuous map
  from $\beta_Q$ to $(Q, \geq)$.
\end{lemma}

\begin{lemma}
  \label{lemma:f:Q}
  Let $\beta$ be a non-limit ordinal, and
  $f \colon \beta \to (\real, \geq)$ be a Scott-continuous map.  The
  image $\img f$ of $f$ is compact in $\Rl$.
\end{lemma}
\begin{proof}
  Since $\beta$ is a non-limit ordinal, it is a dcpo.  Its finite
  elements are the ordinals $\alpha < \beta$ that are not limit
  ordinals, and then it is easy to see that $\beta$ is an algebraic
  domain.  It is also a complete lattice, and every algebraic complete
  lattice is Lawson-compact \cite[Corollary~III.1.11]{gierz03}.

  We claim that $f$ is continuous from $\beta$, with its Lawson
  topology, to $\Rl$.  Let $[a, b[$ be any subbasic open subset of
  $\Rl$.  We aim to show that $f^{-1} ([a, b[)$ is Lawson-open.  If
  $f^{-1} (]-\infty, a[)$ is empty, then
  $f^{-1} ([a, b[) = f^{-1} (]-\infty, b[)$, which is Scott-open since
  $f$ is Scott-continuous; in particular it is Lawson-open.  So let us
  assume that $f^{-1} (]-\infty, a[)$ is non-empty.  Since $\beta$ is
  well-founded, $f^{-1} (]-\infty, a[)$ has a least element $\alpha$.
  Then $f^{-1} (]-\infty, a[) = \upc \alpha$, and therefore
  $f^{-1} ([a, b[) = f^{-1} (]-\infty, b[) \diff f^{-1} (]-\infty, a[)
  = f^{-1} (]-\infty, b[) \diff \upc \alpha$ is Lawson-open.

  Since $\beta$ is compact in the Lawson topology, it follows that
  $\img f$ is compact in $\Rl$.
\end{proof}
In particular, we obtain the following converse to
Lemma~\ref{lemma:Q:limembed}.
\begin{lemma}
  \label{lemma:Q:limembed:op}
  Every well-founded subdcpo $Q$ of $(\real, \geq)$ is compact in $\Rl$.
\end{lemma}
\begin{proof}
  $Q$ is order-isomorphic to a unique ordinal $\beta_Q$ through some
  map $f \colon \beta_Q \to (Q, \geq)$.  Since $Q$ has a least
  element, $\beta_Q$ has a largest element, so $\beta_Q$ cannot be a
  limit ordinal.  Let $D$ be any non-empty family $D$ in $\beta_Q$.
  Since $Q$ is a subdcpo of $(\real, \geq)$, $\inf f (D)$ is an
  element $f (\alpha)$ of $Q$.  Since $f$ is an order-isomorphism,
  $\alpha$ is the supremum of $D$, and therefore $f$ is
  Scott-continuous.  We can now apply Lemma~\ref{lemma:f:Q}.
\end{proof}

\begin{corollary}
  \label{corl:Rl:seq}
  Let $x$ be any real number and $x_0 > x_1 > \cdots > x_n \cdots$ be
  any decreasing sequence of real numbers such that
  $\finf_{n \in \N} x_n = x$.  Then
  $\{x_0, x_1, \cdots, x_n, \cdots\} \cup \{x\}$ is compact in $\Rl$.  \qed
\end{corollary}

Together, Lemmas~\ref{lemma:Q:countable}, \ref{lemma:Q:fQ:limit},
and~\ref{lemma:Q:limembed:op} imply the following.
\begin{theorem}
  \label{thm:Q:Rl}
  The compact subsets of $\Rl$ are exactly the well-founded subdcpos
  of $(\real, \geq)$.  They are all countable.  \qed
\end{theorem}

\begin{remark}
  \label{rem:countable}
  In general, any well-founded chain in $(\real, \geq)$ is countable.
  Conversely, for every countable ordinal $\beta$, there is a
  well-founded chain of $(\real, \geq)$ that is order-isomorphic to
  $\beta$.  This is proved by induction on $\beta$, using the fact
  that every countable ordinal has countable cofinality, namely is the
  supremum of countably many strictly lower countable ordinals.  All
  this is folklore, and is left as an exercise.
\end{remark}

\begin{remark}
  \label{rem:consonant}
  A space $X$ is \emph{consonant} if and only if, for every Scott-open 
  subset~$\mathcal U$ of $\O X$, for every $U \in \mathcal U$, there is a
  compact saturated subset $Q$ of $X$ such that
  $U \in \blacksquare Q \subseteq \mathcal U$, where $\blacksquare Q$
  is the set of open neighborhoods of $Q$.  We now have enough to give
  an elementary proof that $\Rl$ is not consonant
  \cite{Bouziad:Borel,CW:dissonant}.  Let us pick any real number
  $r > 0$.  By Corollary~\ref{corl:lambda},
  $\mathcal U \eqdef \lambda^{-1} (]r, +\infty])$ is a Scott-open
  subset of $\O \Rl$.  For every compact (saturated) subset $Q$ of
  $\Rl$, $Q$ is countable.  Let us write $Q$ as
  $\{x_n \mid n \in I\}$, where $I$ is some subset of $\N$.  Then
  $V \eqdef \bigcup_{n \in I} [x_n, x_n + r/2^{n+1}[$ is an open
  neighborhood of $Q$ such that
  $\lambda (V) \leq \sum_{n \in I} r/2^{n+1} \leq r$, so $V$ is not in
  $\mathcal U$.  It follows that no set of the form $\blacksquare Q$
  is included in $\mathcal U$.
\end{remark}

\subsection{The dcpo $\Smyth\Rl$}
\label{sec:dcpo-smythrl}

Given any topological space $X$, we can form the set $\Smyth X$ of all
non-empty compact saturated subsets of $X$.  $\Smyth X$ is a poset
under the reverse inclusion ordering $\supseteq$ called the
\emph{Smyth powerdomain} of $X$.  When $X$ is well-filtered, this is a
dcpo, where suprema of directed families ${(Q_i)}_{i \in I}$ are their
intersection $\fcap_{i \in I} Q_i$
\cite[Proposition~8.3.25]{goubault13a}.  This is notably the case when
$X = \Rl$, since $\Rl$ is $T_2$, hence sober, hence well-filtered.
Note also that, in that case, every subset is saturated, so we may
safely omit ``saturated'' from the description of elements of
$\Smyth\Rl$.

There are at least two topologies of interest on $\Smyth X$.  One is
the Scott topology on the poset $(\Smyth X, \supseteq)$.  Another one
is the \emph{upper Vietoris topology}, whose basic open sets are the
sets $\Box U \eqdef \{Q \in \Smyth X \mid Q \subseteq U\}$, where $U$
ranges over the open subsets of $X$.  When $X$ is well-filtered (e.g.,
if $X=\Rl$), $\Box U$ is Scott-open, and hence the Scott topology is
finer than the upper Vietoris topology.

The set $\Max \Smyth\Rl$ of maximal points of $\Smyth\Rl$ consists of
the one-element compact sets $\{x\}$, $x \in \real$.  By equating them
with $x$, we equate the set $\Max \Smyth\Rl$ with $\Rl$.  This allows
us to write $\mathcal U \cap \Rl$ for any subset $\mathcal U$ of
$\Smyth\Rl$.
\begin{lemma}
  \label{lemma:calU:open}
  For every Scott-open subset $\mathcal U$ of $\Smyth\Rl$,
  $\mathcal U \cap \Rl$ is open in $\Rl$.
\end{lemma}
\begin{proof}
  Let $x$ be an arbitrary point of $\mathcal U \cap \Rl$.  We claim
  that $\mathcal U \cap \Rl$ contains an interval $[x, x+\epsilon[$
  for some $\epsilon > 0$.  We reason by contradiction, and we assume
  that every interval $[x, x+\epsilon[$ contains a point outside
  $\mathcal U\cap \Rl$.

  We use this to build a sequence of points
  $x_0 > x_1 > \cdots > x_n > \cdots \geq x$ as follows.  Since
  $\mathcal U \cap \Rl$ does not contain $[x, x+1[$, there is a point
  $x_0$ in $[x, x+1[$ that is not in $\mathcal U \in \Rl$.  This
  cannot be $x$, since $x$ is in $\mathcal U \cap \Rl$.  Hence
  $\min (x_0, x+1/2) > x$.  Since $\mathcal U \cap \Rl$ does not
  contain $[x, \min (x_0, x+1/2)[$, there is a point $x_1$ in
  $[x, \min (x_0, x+1/2)[$ that is not in $\mathcal U \in \Rl$.
  Again, $x_1$ is different from $x$.  This allows us to build the
  interval $[x, \min (x_1, x+1/4)[$, and as before, this must contain
  a point $x_2$ outside $\mathcal U \cap \Rl$.  By induction, this
  allows us to define points $x_n$ outside $\mathcal U \cap \Rl$ such
  that $x_{n+1} \in [x, \min (x_n, x+ 1/2^{n+1})[$.  In particular,
  $x_0 > x_1 > \cdots > x_n > \cdots \geq x$.  Also,
  $\finf_{n \in \N} x_n = x$.

  For each $n \in \N$, let $Q_n$ be the set
  $\{x_m \mid m \geq n\} \cup \{x\}$.  This is compact in $\Rl$ by
  Corollary~\ref{corl:Rl:seq}.

  It is clear that $Q_n$ is non-empty, and that $\fcap_{n \in \N}
  Q_n = \{x\}$, which is in $\mathcal U$.  Since $\mathcal U$ is
  Scott-open, some $Q_n$ must be in $\mathcal U$.  This implies that
  $\{x_n\}$, which is included, hence above $Q_n$ in $\Smyth\Rl$, is
  also in $\mathcal U$.  Therefore $x_n$ is in $\mathcal U \cap \Rl$;
  but that is impossible, since all the points $x_n$ were built so as
  to lie outside $\mathcal U \cap \Rl$.
\end{proof}

A \emph{dcpo model} of a $T_1$ space $X$ is a dcpo $P$ such that $\Max
P$, the subset of maximal elements of $P$ with the subspace topology,
from the Scott topology on $P$, is homeomorphic to $X$.  The following
is a special case of Corollary~2.12 of \cite{HLXZ:T1models}, which
says that for every $T_1$, first-countable and well-filtered space,
$\Smyth X$ is a dcpo model of $X$.
\begin{theorem}
  \label{thm:JRl:induced}
  $\Smyth\Rl$ is a dcpo model of $\Rl$.
\end{theorem}
\begin{proof}
  By Lemma~\ref{lemma:calU:open}, every open subset in the subspace
  topology on $\Max \Smyth\Rl$ is open in $\Rl$.  Conversely, for
  every open subset $U$ of $\Rl$, $\Box U$ is a Scott-open subset of
  $\Smyth\Rl$ whose intersection with $\Rl$ is equal to $U$, so $U$ is
  open in the subspace topology on $\Max \Smyth\Rl$.
\end{proof}

 \begin{remark}
 \label{prop:JRl:tops}
 Actually, one can expect more in this case. Xu and Yang\cite{xu21} proved that for a 
 first-countable well-filtered space $X$, in which each compact saturated 
 subset has countable minimal elements (in the specialization order), the Scott and the upper Vietoris 
 topologies coincide on~$\Smyth X$. As we have seen from above that $\Rl$ does satisfy 
 these properties, hence the Scott and the upper Vietoris 
 topologies coincide on~$\Smyth \Rl$. This is slightly stronger than Theorem~\ref{thm:JRl:induced}.
 We speak in full of the development of
 Theorem~\ref{thm:JRl:induced} in order to keep this note self-contained. 
 \end{remark}

\subsection{A continuous, non point-continuous valuation on $\Smyth\Rl$}
\label{sec:continuous-non-point}

By Proposition~\ref{prop:hL:contval}, every Borel measure $\mu$ on
$\Rl$ defines a continuous valuation by restriction to $\O \Rl$, and
we again write that continuous valuation as $\mu$.  By
Theorem~\ref{thm:JRl:induced}, the map $x \mapsto \{x\}$ is a
topological embedding of $\Rl$ into $\Smyth\Rl$.  The image of $\mu$
by that embedding is a continuous valuation $\overline\mu$ on
$\Smyth\Rl$.  Explicitly, we have:
\begin{align}
  \label{eq:mu:bar}
  \overline\mu (\mathcal U) & \eqdef \mu (\mathcal U \cap \Rl)
\end{align}
for every Scott-open subset $\mathcal U$ of $\Smyth\Rl$.

\begin{theorem}
  \label{thm:mu:bar:notpc}
  Let $\mu$ be any Borel measure on $\Rl$ with the property that there
  is an open subset $U$ of $\Rl$ such that $0 < \mu (U) < +\infty$,
  and $\mu (\{x\})=0$ for every $x \in U$.  Then the continuous
  valuation $\overline\mu$ is not point-continuous.  In particular,
  $\overline\lambda$ is not point-continuous.
\end{theorem}
\begin{proof}
  Since $\mu (U) > 0$ and since $\overline\mu (\Box U) = \mu (U)$, we
  can find a real number $r$ such that $0 < r < \overline\mu (\Box U)$.
  
  We will show that, for every finite subset
  $A \eqdef \{Q_1, \cdots, Q_m\}$ of $\Box U$, there is an open
  neighborhood $\mathcal V$ of $A$ such that
  $\overline\mu (\mathcal V) \leq r$.

  We make the following preliminary claim $(*)$: for every $x \in U$, for
  every $a > 0$, there is an $\epsilon > 0$ such that
  $[x, x+\epsilon[ \subseteq U$ and $\mu ([x, x+\epsilon[) \leq a$.
  Indeed, it is a general property of measures that
  $\mu (\fcap_{n \in \N} E_n) = \finf_{n \in \N} \mu (E_n)$ for
  any decreasing family of measurable subsets $E_n$ such that
  $\mu (E_n) < +\infty$ for at least one $n$.  Hence
  $\mu (\{x\}) = \mu (\fcap_{n \in \N} [x, x+\epsilon_0/2^n[) =
  \finf_{n \in \N} \mu ([x, x+\epsilon_0/2^n[)$, where
  $\epsilon_0 > 0$ is chosen so that $[x, x+\epsilon_0[ \subseteq U$.
  Since $\mu (\{x\})=0$, there is an $n \in \N$ such that
  $\mu ([x, x+\epsilon_0/2^n[ < a$.

  Let $s > 0$ be such that $ms \leq r$.  For every
  $i \in \{1, \cdots, m\}$, $Q_i$ is countable
  (Lemma~\ref{lemma:Q:countable}), so let us write it as
  $\{x_{i0}, x_{i1}, \cdots\}$.  (We allow for infinite repetitions of
  elements in order not to have to make a special case when $Q_i$ is
  finite.)  For every $i \in \{1, \cdots, m\}$, for every $j \in \N$,
  we use $(*)$ to find a number $\epsilon_{ij} > 0$ such that
  $\mu ([x_{ij}, x_{ij}+\epsilon_{ij}[) \leq s/2^{j+1}$.  Let
  $V_i \eqdef \bigcup_{j \in \N} [x_{ij}, x_{ij}+\epsilon_{ij}[$.  We
  note that
  $\mu (V_i) \leq \sum_{j \in \N} \mu ([x_{ij}, x_{ij}+\epsilon_{ij}[)
  = \sum_{j \in \N} s/2^{j+1} = s$.  We now define $V$ as
  $\bigcup_{i=1}^m V_i$.  Then
  $\mu (V) \leq \sum_{i=1}^m \mu (V_i) \leq ms \leq r$.

  Clearly, $Q_i$ is included in $V_i$, hence in $V$, for every
  $i \in \{1, \cdots, m\}$, so $A$ is included in $\Box V$.  We define
  $\mathcal V$ as $\Box V$.  Then
  $\overline\mu (\mathcal V) = \mu (V) \leq r$.
\end{proof}

\subsection{More remarks on $\Smyth\Rl$}
\label{sec:concluding-remarks}

In \cite[Theorem 3.1]{lyu19}, Lyu and the second author showed that a space 
$X$ is locally compact if and only if that $\Smyth X$ is core-compact
in the upper Vietoris topology. It is easy to see that $\Rl$ is not locally 
compact as the interior of each compact set is empty. So $\Smyth\Rl$
is not core-compact in the upper Vietoris topology.  By Remark~\ref{prop:JRl:tops}, we have

\begin{proposition}
  \label{prop:QRl:notcc}
  $\Smyth\Rl$ is not core-compact in its Scott topology.
\end{proposition}

$\Rl$ exhibits a diverse mix of pleasant and unpleasant properties,
and so does $\Smyth\Rl$.  While Proposition~\ref{prop:QRl:notcc} would
be on the unpleasant side, the following shows more regularity.

\begin{proposition}
  \label{prop:QRl:sober}
  $\Smyth\Rl$ is sober.
\end{proposition}
\begin{proof}
  Theorem~3.13 of \cite{heckmann13} states that, for any topological
  space $X$, $X$ is sober if and only if $\Smyth X$ is sober in the
  upper Vietoris topology.  Since $\Rl$ is $T_2$, it is sober.  By
  Remark~\ref{prop:JRl:tops}, the upper Vietoris topology
  coincides with the Scott topology on $\Smyth\Rl$, so $\Smyth\Rl$ is
  sober.
\end{proof}

In particular, $\Smyth\Rl$ is well-filtered, something we can rederive
in another way.
\begin{proposition}
  \label{prop:QRl:nice}
  For every well-filtered, coherent space, $\Smyth X$ is a
  meet-continuous dcpo inf-semilattice, which is 
  well-filtered and coherent in its Scott topology.  Hence
  $\Smyth\Rl$, $\Smyth^2\Rl$, \ldots{} all are meet-continuous dcpo
  inf-semilattices, which are well-filtered and coherent in their
  Scott topologies.
\end{proposition}
\begin{proof}
  Let $X$ be a well-filtered, coherent space.  In a well-filtered
  space, filtered intersections $\fcap_{i \in I} Q_i$ of compact
  saturated subsets are compact saturated, and are non-empty if all
  the sets $Q_i$ are non-empty.  Hence $\Smyth X$ is a dcpo.
  
  Given any two elements $Q$ and $Q'$ of $\Smyth X$, their infimum is
  $Q \cup Q'$.  Since directed suprema are filtered intersections, and
  intersections commute with binary unions, $\Smyth\Rl$ is a
  meet-continuous inf-semilattice.
 
  For every dcpo $P$, let its \emph{lifting} $P_\bot$ be $P$ plus a
  fresh element $\bot$ below all elements of $P$.  It is an easy
  exercise to show that the Scott-open subsets of $P_\bot$ are those
  of $P$, plus $P_\bot$ itself, and that the compact saturated subsets
  of $P_\bot$ are those of $P$ plus $P_\bot$.  It follows that $P$ is
  well-filtered if and only if $P_\bot$ is.

  A bounded-complete dcpo is one in which every (upper) bounded family
  has a least upper bound, or equivalently in which every non-empty
  family has a greatest lower bound.  $\Smyth X$ is not
  bounded-complete in general, since the empty family has no least
  upper bound unless $X$ is compact.

  However, ${(\Smyth X)}_\bot$ is bounded-complete: the least upper
  bound of the empty set is $\bot$, and the least upper bound of any
  non-empty set $A$ bounded by some element $Q_0$ of $\Smyth X$ is
  $\bigcap A$, which is compact because $\Rl$ is $T_2$, and non-empty
  because it contains $Q_0$.

  Corollary~3.2 of \cite{xi17} shows that every bounded-complete dcpo
  is well-filtered in its Scott topology.  Hence ${(\Smyth X)}_\bot$ is
  well-filtered, and therefore $\Smyth X$ is well-filtered as well.

  Lemma~3.1 of \cite{jia16a} states that any well-filtered dcpo $X$ in which
  $\upc x \cap \upc y$ is compact saturated for all $x, y \in X$ is
  coherent.  For all $Q, Q' \in \Smyth X$,
  $\upc Q \cap \upc Q' = \upc (Q \cap Q')$, and $Q \cap Q'$ is again
  in $\Smyth X$ since $X$ is coherent.  In particular,
  $\upc Q \cap \upc Q'$ is a compact saturated subset of $\Smyth X$.
  Therefore $\Smyth X$ is coherent.

  Finally, $\Rl$ is $T_2$, hence (sober hence) well-filtered, and
  coherent, so we may apply the above to $X \eqdef \Rl$, then to $X
  \eqdef \Smyth \Rl$, and so on.
\end{proof}

\section{Concluding Remarks}

We have given two concrete examples (Corollary~\ref{coro:muisnotpointcontinuous}, Thorem~\ref{thm:mu:bar:notpc}) on dcpo's to separate minimal valuations, point-continuous valuations
and continuous valuations, showing these three classes of valuations differ from each other.

In \cite{monad-m}, the Fubini-type equation
\begin{equation}
\label{eq:fubini}
\int_{x\in D} \int_{y\in E} h(x, y) d\nu d\xi  = \int_{y\in E} \int_{x\in D} h(x, y) d\xi d\nu
\end{equation}
is established when either $\nu$ or $\xi$ is point-continuous, where $D$ and $E$
are dcpo's and $h\colon D\times E\to \R$ are Scott-continuous. (For a definition of the integration, see
\cite{jones90}.)
 This is crucial in
proving that Heckmann's point-continuous valuations monad is commutative over $\dcpo$. 
However, it is unknown whether the Equation~(\ref{eq:fubini}) holds for general continuous valuations $\nu$ 
and $\xi$, a crucial question in establishing commutativity of the valuations moand $\VV$ on $\dcpo$.
The aforementioned result in \cite{monad-m} entails that any valuations $\nu$ and $\xi$ that possibly fail Equation~(\ref{eq:fubini})
must also fail to be point-continuous. 
Hence if one aims to find examples on dcpo's to invalidate Equation~(\ref{eq:fubini}), the valuations in desire must be non-point-continuous valuations. 
For the first time, we have given a continuous valuation $\overline \lambda$ that is not point-continuous on dcpo's, but more non-point-continuous valuations (of different types from those in Theorem~\ref{thm:mu:bar:notpc}) are needed before they are sent to test Equation~(\ref{eq:fubini}).

\section*{Acknowledgement}

The second author acknowledges support from NSFC (No. 12001181), and he would also like to thank Andre Kornell and Michael Mislove
for useful discussions.


\begin{thebibliography}{10}

\bibitem{abramsky94}
S.~Abramsky and A.~Jung.
\newblock Domain theory.
\newblock In S.~Abramsky, D.~M. Gabbay, and T.~S.~E. Maibaum, editors, {\em
  Semantic Structures}, volume~3 of {\em Handbook of Logic in Computer
  Science}, pages 1--168. Clarendon Press, 1994.

\bibitem{Adamski:measures}
W.~Adamski.
\newblock {$\tau$}-smooth {B}orel measures on topological spaces.
\newblock {\em Mathematische Nachrichten}, 78:97--107, 1977.

\bibitem{alvarez99}
M.~Alvarez-Manilla, A.~Edalat, and N.~Saheb-Djahromi.
\newblock An extension result for continuous valuations.
\newblock {\em Journal of the London Mathematical Society}, 61:629--640, 2000.

\bibitem{Bouziad:Borel}
A.~Bouziad.
\newblock Borel measures in consonant spaces.
\newblock {\em Topology and its Applications}, 70:125--138, 1996.

\bibitem{CW:dissonant}
C.~Costantini and S.~Watson.
\newblock On the dissonance of some metrizable spaces.
\newblock {\em Topology and its Applications}, 84:259--268, 1996.

\bibitem{gierz03}
G.~Gierz, K.~H. Hofmann, K.~Keimel, J.~D. Lawson, M.~Mislove, and D.~S. Scott.
\newblock {\em Continuous Lattices and Domains}, volume~93 of {\em Encyclopedia
  of Mathematics and its Applications}.
\newblock Cambridge University Press, 2003.

\bibitem{goubault13}
J.~Goubault-Larrecq.
\newblock Full abstraction for non-deterministic and probabilistic extensions
  of~{PCF}~{I} --- the angelic cases.
\newblock {\em Journal of Logical and Algebraic Methods in Programming}, 84(1):155--184, 2015.

\bibitem{goubault13a}
J.~Goubault-Larrecq.
\newblock {\em Non-Hausdorff Topology and Domain Theory}, volume~22 of {\em New
  Mathematical Monographs}.
\newblock Cambridge University Press, 2013.

\bibitem{goubault21}
J.~Goubault-Larrecq.
\newblock Products and projective limits of continuous valuations on {$T_{0}$}
  spaces.
\newblock {\em Accepted for publication in ``{M}athematical {S}tructure in
  {C}omputer {S}cience''}, 2021.

\bibitem{HLXZ:T1models}
Q.~He, G.~Li, X.~Xi, and D.~Zhao.
\newblock Some results on poset models consisting of compact saturated subsets.
\newblock {\em Electronic Notes in Theoretical Computer Science}, 345:77--85,
  2019.
\newblock Proceedings of the 8th International Symposium of Domain Theory
  (ISDT'19), A. Jung, Q. Li, L. Xu and G.-Q. Zhang, editors.

\bibitem{heckmann95}
R.~Heckmann.
\newblock Spaces of valuations.
\newblock Technical Report A 09/95, FB 14 Informatik, Universit\"at des
  Saarlandes, 66041 Saarbr\"ucken, Germany, 1995.

\bibitem{heckmann96}
R.~Heckmann.
\newblock Spaces of valuations.
\newblock In S.~Andima, R.~C. Flagg, G.~Itzkowitz, P.~Misra, Y.~Kong, and
  R.~Kopperman, editors, {\em Papers on General Topology and Applications:
  Eleventh Summer Conference at the University of Southern Maine}, volume 806
  of {\em Annals of the New York Academy of Sciences}, pages 174--200, 1996.

\bibitem{heckmann13}
R.~Heckmann and K.~Keimel.
\newblock Quasicontinuous domains and the {S}myth powerdomain.
\newblock In D.~Kozen and M.~Mislove, editors, {\em Proceedings of the 29th
  Conference on the Mathematical Foundations of Programming Semantics}, volume
  298 of {\em Electronic Notes in Theoretical Computer Science}, pages
  215--232. Elsevier Science Publishers {B.V.}, 2013.
  
 
  
\bibitem{ho18}
W.~K.~Ho and J.~Goubault-Larrecq and A.~Jung and X.~Xi.
\newblock The {H}o-{Z}hao {P}roblem.
\newblock {\em Logical Methods in Computer Science}, 14(1), 2018.
  
  
\bibitem{isbell82}
J. Isbell.
\newblock Completion of a construction of {Johnstone}.
\newblock {\em Proceedings of the American Mathematical Society}, 85:333--334, 1982.
  
\bibitem{jia16a}
X.~Jia, A.~Jung, and Q.~Li.
\newblock A note on coherence of dcpos.
\newblock {\em Topology and its Applications}, 209:235--238, 2016.

\bibitem{monad-m}
X.~Jia, B.~Lindenhovius, M.~Mislove, and V.~Zamdzhiev.
\newblock Commutative monads for probabilistic programming languages.
\newblock In {\em Logic in Computer Science ({LICS 2021})}, 2021.

\bibitem{jones90}
C.~Jones.
\newblock {\em Probabilistic Non-Determinism}.
\newblock PhD thesis, University of Edinburgh, Edinburgh, 1990.
\newblock Also published as Technical Report No.~CST-63-90.

\bibitem{jones89}
C.~Jones and G.~Plotkin.
\newblock A probabilistic powerdomain of evaluations.
\newblock In {\em Proceedings of the 4th Annual Symposium on Logic in Computer
  Science}, pages 186--195. IEEE Computer Society Press, 1989.
  
  
\bibitem{johnstone81}
P.~T.~Johnstone.
\newblock Scott is not always sober.
\newblock In {\em Continuous Lattices, Proceedings Bremen}, 871:282--283, 1981. 
 

\bibitem{lyu19}
Z.~Lyu and X.~Jia.
\newblock Core-compactness of smyth powerspaces.
\newblock Available at \url{https://arxiv.org/abs/1907.04715}, July 2019.

\bibitem{Royden88}
H.~L. Royden.
\newblock {\em Real {A}nalysis}.
\newblock Macmillan, New York, 3rd edition, 1988.

\bibitem{Sorgenfrey}
R.~H.~Sorgenfrey.
\newblock On the topological product of paracompact spaces.
\newblock {\em Bulletin of the American Mathematical Society}, 53:631--632,
  1947.

\bibitem{tix95}
R.~Tix.
\newblock Stetige {B}ewertungen auf topologischen {R}{\"a}umen.
\newblock Master's thesis, Technische Hochschule Darmstadt, June 1995.
\newblock 51pp.

\bibitem{xi17}
X.~Xi and J.~Lawson.
\newblock On well-filtered spaces and ordered sets.
\newblock {\em Topology and its Applications}, 228:139--144, September 2017.

\bibitem{xu21}
X.~Xu and Z.~Yang.
\newblock Coincidence of the upper vietoris topology and the scott topology.
\newblock {\em Topology and its Applications}, 288, 107480, December 2021.

\end{thebibliography}
\end{document}